\documentclass[journal]{IEEEtran}

\IEEEoverridecommandlockouts

\usepackage{relsize}
\usepackage{moreverb}
\usepackage{mathrsfs}
\usepackage{amsmath}
\usepackage{amssymb}
\usepackage{stfloats}
\usepackage{graphicx}
\usepackage{makecell}
\usepackage{xcolor}
\usepackage{cite}

\usepackage{algorithm}
\usepackage{algorithmicx}
\usepackage{algpseudocode}

\usepackage{multirow}       
\usepackage{booktabs}
\usepackage{multirow}
\usepackage{tabularx,ragged2e,booktabs,caption}

\ifCLASSOPTIONcompsoc
\usepackage[tight,normalsize,sf,SF]{subfigure}
\else
\usepackage[tight,footnotesize]{subfigure}
\usepackage{subfigure}

\usepackage{amsmath}
\usepackage{amsthm}

\newtheorem{theorem}{Theorem}

\usepackage{setspace}
\usepackage[hidelinks]{hyperref}
\usepackage{graphicx}
\usepackage{flushend}
\allowdisplaybreaks[2]

\begin{document}
	
	\title{\huge{Fairness vs. Equality: RSMA-Based Multi-Target and Multi-User Integrated Sensing and Communications}
	\thanks{This work was supported by National Key Laboratory of Unmanned Aerial Vehicle Technology in NPU (Grant No. WR202404), Fundamental Research Funds for the Central Universities (Grant No. G2024WD0159, D5000240239), 
	National Key Laboratory Fund Project for Space Microwave Communication (Grant No. HTKI2024KL504010), 
	2023 Zhongdian Tian’ao Innovation Theory and Technology Group Fund (Grant No. 2023JSQ0101), 
	and Guangdong Basic and Applied Basic Research Foundation (No. 2021A1515110077). 
	The work of A.-A. A. Boulogeorgos is supported by MINOAS. His research project MINOAS is implemented in the framework of H.F.R.I call “Basic research Financing (Horizontal support of all Sciences)” under the National Recovery and Resilience Plan “Greece 2.0” funded by the European Union –NextGenerationEU (H.F.R.I. Project Number: 15857).” \emph{(Corresponding author: Rugui Yao.)}}
	%
	
	\thanks{Xudong Li is with the China Ship Development and Design Center, Wuhan 430064, China (e-mail: xudong\_good@mail.nwpu.edu.cn).}
	\thanks{Rugui Yao is with the School of Electronics and Information, Northwestern Polytechnical University, Xi’an 710072, China (e-mail: yaorg@nwpu.edu.cn).}
	\thanks{Alexandros-Apostolos A. Boulogeorgos is with the Department of Electrical and Computer Engineering, University of Western Macedonia, Kozani 50100, Greece (email: al.boulogeorgos@ieee.org).}	
	\thanks{Theodoros A. Tsiftsis is with Department of Informatics Telecommunications, University of Thessaly, Lamia 35100, Greece, and also with the Department of Electrical and Electronic Engineering, University of Nottingham Ningbo China, Ningbo 315100, China (e-mail: tsiftsis@uth.gr).}

}

{\Large{\author{\IEEEauthorblockN{Xudong Li, Rugui Yao,~\IEEEmembership{Senior Member,~IEEE}, 
				Alexandros-Apostolos A. Boulogeorgos,~\IEEEmembership{Senior Member,~IEEE},
				and Theodoros A. Tsiftsis,~\IEEEmembership{Senior Member,~IEEE} }}}}

	\maketitle
	
\begin{abstract}
Equality lies in undifferentiated allocation, disregarding target disparities or contextual needs, while fairness centers on differentiated compensation, explicitly addressing systemic disadvantages faced by marginalized targets. Existing works on multi-objective sensing optimization in ISAC systems have predominantly emphasized equality, but fairness concerns remain inadequately addressed.
This paper investigates the tradeoff between sensing and communication in an \emph{integrated sensing and communications} (ISAC) system comprising multiple sensing targets and communication users. A dual-functional base station conducts downlink data transmission services based on \emph{rate-splitting multiple access} (RSMA) for multiple users, while sensing surrounding multiple targets. To enable effective multicast communications and ensure fair and balanced multi-target sensing and under a constrained power budget, we propose a multi-target sensing enhancement scheme incorporating fairness-aware \emph{beamforming} (BF), common rate splitting, and sensing power allocation. The proposed scheme minimizes the sensing \emph{Cramér-Rao bound} (CRB), while maximizing communication rate demands. Specifically, we derive closed-form expressions for both sensing CRB and communication rates. Building upon them, we formulate an optimization problem aiming to minimize the sensing CRB, while maximizing the communication rates. Considering the non-convex nature of the original optimization problem poses significant computational challenges, we transform the tradeoff optimization into a Pareto-optimal  problem by employing Taylor series expansion, semi-definite relaxation, successive convex approximation, and penalty function to transform the non-convex problem and associated constraints into tractable forms. Extensive simulations validate the theoretical analysis and demonstrate significant advantages of the proposed RSMA-based fairness-aware BF over non-orthogonal multiple access, space division multiple access, and orthogonal multiple access, through comprehensive comparisons in two key aspects: CRB performance improvement and sensing-communication tradeoff characteristics. The proposed optimization framework exhibits remarkable superiority in enhancing both sensing accuracy and communication quality for ISAC systems.
\end{abstract}


\IEEEpeerreviewmaketitle

\section{Introduction}

\subsection{Background}
\IEEEPARstart{B}{y} 2030, \emph{sixth generation} (6G) wireless networks will enable holographic communication \cite{holo} and digital twins, facilitating the comprehensive interaction between massive perceptive and integrated agents, bridging the virtual and physical worlds \cite{6G_ISAC}. 
In urban traffic, drones, robot clusters, and other practical application scenarios, the mode of information interaction is no longer limited to information transmission, but extends to information sensing and information calculation, and target positioning, environment sensing and communication capabilities have evolved into fundamental requirements that cannot be ignored. Therefore, \emph{integrated sensing and communication} (ISAC) has been incorporated into the technical field of 6G network \cite{33_ISAC_re}. By harmonizing sensing and communication functionalities within a unified framework, ISAC promises enhanced spectrum efficiency \cite{5_ISAC,7_ISAC}, reduced hardware costs, and synergistic performance gains. However, in scenarios involving multiple communication users and diverse sensing targets \cite{LXD_GREAT_all}, critical challenges emerge: Resource competition between sensing waveform design and multi-user interference management often leads to performance tradeoffs.
Given these problems, ISAC becomes an emerging technology with important fundamental research and applicability value \cite{Survey_ISCC,Survey_secure_ISAC}, motivating an attention steering towards ISAC performance analysis and improvement researches. 

\subsection{Literature Review}
The authors of \cite{8_ISAC} analyzed the interplay between sensing and communication in the ISACs, explored multiple performance tradeoffs, and identified the potential integration of ISAC with other emerging communication technologies. In \cite{10_ISAC}, a closed-form expression was obtained for data rate triggered sensing-control pattern activation design, where both data rate requirement in millimeter wave/terahertz communications and motion control performance of \emph{unmanned aerial vehicle} (UAV) were guaranteed to realize the data transmission from a UAV to a ground \emph{base station} (BS). 
The authors of \cite{12_ISAC} quantitatively described the performance limitations and tradeoffs between sensing and communication in distributed ISAC networks with transmit power and bandwidth budgets given. In \cite{13_ISAC_0325}, the authors developed an ISAC framework based on the Markov decision processes and the deep reinforcement learning to optimize \emph{beamforming} (BF) to improve sensing and communication performance in the dynamic and uncertain environment. A full-duplex ISAC scheme was presented in \cite{14_ISAC}, which utilized the waiting time of the traditional pulse radar to transmit communication signals, improve the communication rate, and overcome problems of sensing overlap and near-target blind distance.
All the aforementioned works agree on the fact that in conventional ISAC systems sensing accuracy and effective provision of communication services were seriously hindered by multiplicative fading, unprocessed mutual interference, and the inadaptability of static scheduling modes in changeable and dynamic scenarios.

In view of the above challenges, scholars and experts have made notable contributions to the interference management domain. In more detail, in \cite{15_ISAC}, a \emph{space division multiple access} (SDMA) scheme was designed to suppress interference, which employs a linear precoding to distinguish users in the spatial domain, relying entirely on treating any residual multi-user interference as noise. The authors of \cite{16_ISAC} investigated SDMA-based communication performance gain of the ISAC system. In \cite{OMA_ISAC},  the authors evaluated the performance of an \emph{orthogonal multiple access} (OMA)-based Semi-ISaC network.
In contrast to SDMA and OMA, \emph{non-orthogonal multiple access} (NOMA) operation principle lies on stacking encoding at the transmitting end and \emph{successive interference cancellation} (SIC) encoding at the receiving end. In this scheme, users are superimposed in the power domain, and the users with better channel conditions are forced to fully decode and eliminate interference generated by other users through user grouping and sorting \cite{15_ISAC}. 
Given the joint optimization of sensing and communication based on the NOMA, a NOMA-ISAC solution for uplink transmission was given in \cite{17_ISAC} to reduce the mutual interference between radar signals and communication signals. 
The authors of \cite{19_ISAC} presented a joint optimization scheme of the BF, NOMA transmission duration, and target sensing scheduling to maximize the sensing efficiency of ISAC systems, while ensuring a high-level communication quality. An ISAC iterative channel estimation method was articulated in \cite{20_ISAC}, which realized uplink transmission through multi-channel estimation of received non-orthogonal communication signals and sensing signals, and enhanced the spectrum efficiency while both of sensing performance and communication performance were ensured. The authors of \cite{21_ISAC} reported a scheme based on the BF to maximize the weighted sum of the communication throughput and effective sensing power, so as to jointly enhance the sensing performance and communication performance of the NOMA-ISAC system. An ISAC scheme based on the multi-domain NOMA were designed in \cite{30_ISAC}, which transmitted data streams in parallel and non-orthogonally in the time-frequency domain and the delayed Doppler domain with potential targets being detected.

Though several studies on SDMA, OMA, and NOMA that are employed to improve the ISAC system performance are carried out, it is emphasized that these conventional \emph{multiple access} (MA) architectures represented by SDMA, OMA, and NOMA fall into extreme situations. Specifically, SDMA completely regards interference as noise, and thus the reliability of ISAC system seriously deteriorates, and OMA exhibits inherent limitations in balancing heterogeneous \emph{quality-of-service} (QoS) requirements. On the contrary, NOMA decodes interference one by one, which implies that the effectiveness of the ISAC system is hard to guarantee. Therefore, the above works have the sub-optimal interference management performances in the ISAC system. 

Compared with SDMA, OMA, and NOMA, recent advances in \emph{rate-splitting multiple access} (RSMA) offer a paradigm shift. By decomposing user messages into common and private streams, RSMA enables flexible interference management through superposition coding and SIC, fully absorbing the respective advantages of SDMA, OMA, and NOMA. Then, both high reliability and high effectiveness are achieved \cite{6_ISAC}. 
Considering the constraints of rate requirements and transmitting power budget, the authors of \cite{24_ISAC} designed the RSMA structure and parameters to minimize the CRB of the sensing response matrix at the radar receiver, and optimize the sensing performance under the constraint of demand-satisfying communication service.
An indicative  example of an RSMA assisted ISAC waveform design was documented in \cite{25_ISAC}, which aimed to jointly minimize the CRB of the target detection and maximize the minimum fairness rate among communication users under the power constraints of a single transmitting antenna. Considering the power consumption of low-resolution digital-to-analog converters on each \emph{radio frequency} (RF) chain under communication performance and sensing performance constraints, the authors of \cite{26_ISAC} found the optimal numbers of pre-encoders and active RF chains to maximize the energy efficiency of an RSMA-ISAC system. An ISAC system based on uplink RSMA was presented in \cite{27_ISAC}. With the transmit BF and the receive BF jointly optimized, radar sensing performance was improved to the maximum extent while the communication throughput requirement of each user, and the transmitting power constraints of BS and communication users were met. In the scenario of a multi-antenna ISAC satellite system, under the constraints of QoS and the transmitting power budget per feed, RSMA scheme combined with the BF optimization was taken into consideration to fulfill a higher sensing accuracy \cite{28_ISAC}. 

The authors of \cite{RSMA_ISAC_Sensing_N} introduced general RSMA-assisted ISAC architecture with multiple sensing targets, where the sensing CRB was minimized as well as the minimum fairness rate was maximized jointly. In \cite{MT_ISAC1},  the performance tradeoff for a multi-target ISAC was studied. The authors of \cite{MT_ISAC_2} documented a power allocation algorithm to maximize the sensing SINR while ensuring minimum
communication requirements. In \cite{MT_ISAC_3}, the waveform design
problem in a downlink multi-user and multi-target ISAC system
under different communication-sensing performance preferences. Ouyang \emph{et al.} presented a BF method that sensed multiple
moving targets, while communicating with multiple users \cite{MT_ISAC_4}. The authors of \cite{MT_ISAC_5}  investigated transmit BF for MIMO-ISAC systems in scenarios with multiple radar targets and communication users.


\subsection{Motivation and Contribution}
Recent advances in the RSMA-assisted ISAC have demonstrated promising results in balancing sensing and communication functionalities \cite{24_ISAC,25_ISAC,26_ISAC,27_ISAC,28_ISAC,RSMA_ISAC_Sensing_N}. However, four critical limitations persist:

1) Existing studies on RSMA-assisted ISAC systems \cite{24_ISAC,25_ISAC,26_ISAC,27_ISAC,28_ISAC}  predominantly focus on single-target scenarios, thereby ignoring the inter-target CRB coupling in multi-target scenarios. This simplification leads to optimistic performance predictions.

2) Researches on RSMA-assisted ISAC systems with multiple sensing targets remain scarce \cite{RSMA_ISAC_Sensing_N}. Even in more generalized multi-target ISAC frameworks, existing works predominantly focus on aggregate sensing performance metrics, overlooking further analysis of the sensing outcomes obtained for individual targets or developing mechanisms to balance sensing accuracy across heterogeneous targets \cite{MT_ISAC1,MT_ISAC_2,MT_ISAC_3,MT_ISAC_4,MT_ISAC_5}.

3) Comparative evaluations of RSMA-assisted ISAC systems with respect to sensing performance and sensing-communication tradeoffs under different MA schemes, such as RSMA, NOMA, SDMA, and OMA remain insufficient and underexplored \cite{24_ISAC,25_ISAC,26_ISAC,27_ISAC,28_ISAC,RSMA_ISAC_Sensing_N,MT_ISAC1,MT_ISAC_2,MT_ISAC_3,MT_ISAC_4,MT_ISAC_5}.

4) Studies on RSMA-assisted ISAC systems or even more generalized ISAC systems frequently assume that the user's \emph{channel state information} (CSI) is perfect. However, in realistic environments, channel characteristics inevitably vary due to dynamic factors such as mobility, multipath fading, and interference. These conditions make the continuous and accurate acquisition of CSI particularly challenging. 

Motivated by the above, this paper proposes an RSMA-assisted, fairness-aware performance enhancement framework for multi-user multi-target ISAC systems. The framework integrates the joint optimization of beamforming, power allocation, and common rate splitting, while accounting for imperfect CSI at the \emph{dual-function base station} (DFBS), to simultaneously enhance both sensing and communication capabilities.
Pointedly, the key contributions are threefold:

\begin{enumerate} 
	\item Unified fairness-aware sensing framework: We develop a new framework for multi-target fair sensing, which can roughly judge the realisable sensing CRB of different targets through initial estimation, and narrow the differences of sensing CRBs among different targets through fairness-aware BF, sensing power allocation, and common rate splitting with the imperfect users' CSI.
	
	\item Joint sensing-communication optimization: A non-convex optimization problem is formulated to minimize the CRBs of multiple targets and maximize communication rates, subject to constraints of rates of common and private streams, power budgets, and echo signal SCNR. This formulation explicitly addresses the coupling of fair-sensing-oriented BF vectors, sensing power allocation and common rate splitting.
	
	\item Algorithmic framework with theoretical guarantees: To solve the formulated problem efficiently, a computational framework is proposed based on Taylor series expansion, successive convex approximation (SCA), semi-definite relaxation (SDR), and penalty functions. This approach transforms the non-convex problem into a tractable form while preserving theoretical convergence guarantees. Numerical results demonstrate that, compared to ISAC benchmark schemes based on NOMA, SDMA, OMA, the proposed CRB based on RSMA is reduced by at least 29.14\% while maintaining comparable or even superior communication performance.

\end{enumerate}
\begin{figure}[t]
	\centering
	\includegraphics[width=7.498cm,height=6.345cm]{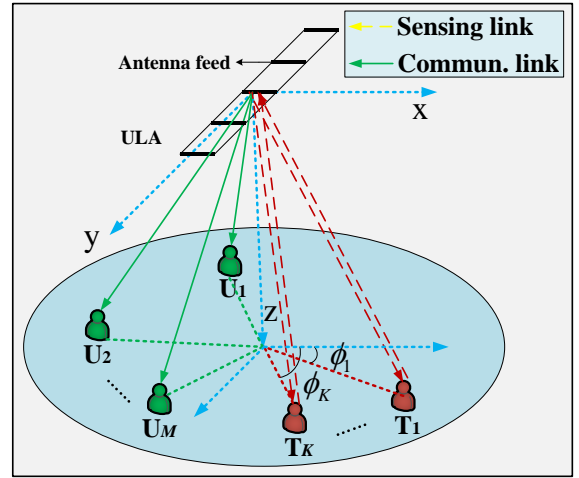}
	\caption{Considered multi-user multi-target ISAC system.}\label{fig_system_model} 
\end{figure}
\subsection{Organization and Notation}
The remainder of this paper is organized as follows. The multi-user multi-target ISAC system architecture, sensing and communication signal models and performance indicators of sensing and communications in the ISAC system are introduced in Section II. Section III formulates the optimization problem, and reports the corresponding solution. Numerical results and simulations are provided in Section IV. Finally, Section V concludes this paper by summarizing the main message and key remarks.

\textbf{Notation:} Matrices and vectors are denoted as uppercase boldface and lowercase boldface, respectively. ${\left(  \cdot  \right)^{\rm{H}}}$, ${\mathop{\rm Tr}\nolimits} \left(  \cdot  \right)$, ${\mathop{\rm rank}\nolimits} \left(  \cdot  \right)$, and ${\left\|  \cdot  \right\|_2}$ are the Hermitian transpose operation, trace operation, rank operation, and 2-norm operation. ${\mathbb{C}^{x \times y}}$ stands for the 2-dimension complex space. ${{\bf{I}}_{Z \times Z}}$ represents the $Z \times Z$ identity matrix, ${{\bf{I}}_{{Z \times 1}}}$ represents the $\left( {{Z} \times 1} \right)$-dimension unit vector, $ \otimes $ and $\odot$ denote the Kroneker and Hadamard products of two matrices, and $\mathcal{O}$ is computational complexity, respectively. Besides, the vectorization of a matrix ${\bf{A}} \in {\mathbb{R}^{m \times n}}$, denoted as vec(\textbf{A}), is a linear transformation that stacks the elements of the matrix ${\bf{A}} = \left[ {{a_{ij}}} \right]$ column-wise into an $mn \times 1$ vector. Formally, ${\mathop{\rm vec}\nolimits} \left( {\bf{A}} \right) = {\left[ {{a_1},{a_2}, \cdots ,{a_n}} \right]^{\rm{T}}} \in {\mathbb{R}^{mn \times 1}}$, where ${a_i} \in {\mathbb{R}^{m \times 1}}$ represents the $i$-th column of \textbf{A}.

\section{System and Channel Model}

\subsection{Network topology}
In the ISAC system illustrated in Fig. \ref{fig_system_model}, a DFBS equipped with a \emph{uniform linear array} (ULA) comprising $N_{\rm t}$ transmit antennas and $N_{\rm r}$ receive antennas, provides RSMA-based downlink multicast communication services to $M$ single-antenna communication users. The DFBS senses the angles and complex parameters of $K$ targets.

Herein, for communications, we assume that the DFBS possesses coarse knowledge of users, the accuracy of which is contingent on the uncertainty of the users’ CSI ascertained by the DFBS. The smaller the CSI uncertainty, the more precise the DFBS’s beampattern gets, enabling more targeted communication services under limited power budgets and thereby enhancing the communication performance of the ISAC system. On the other hand, for sensing, the DFBS can calculate the positions of $K$ targets with the estimated angles and complex parameters.

\subsection{Communication Signal Model}
Since the RSMA-based multicast communication scheme is employed, transmitted messages are partitioned into two parts. The first part is the common stream, while the second part is the private stream. Let $\left\{ {{z_{i,{\rm{c}}}},{z_{i,{\rm{p}}}}} \right\},i \in \left\{ {1,2, \cdots ,M} \right\}$ denote the common and private streams. Common stream $\left\{ {{z_{{\rm{1,c}}}},{z_{{\rm{2,c}}}}, \cdots {z_{M{\rm{,c}}}}} \right\}$ is combined and encoded as the $s_{\rm c}$ via a codebook.
Private stream $\left\{ {{z_{{\rm{1,p}}}},{z_{{\rm{2,p}}}}, \cdots {z_{M{\rm{,p}}}}} \right\}$ is encoded separately as $\left\{ {{s_1,{s_2}, \cdots ,{s_ M}}} \right\}$. Herein, signals like $s_{\rm c}$ and ${s_{ m}}$ are unrelated with the expectation and variance respectively being 0 and 1.
Then, ${{\bf{u}}_{\rm{c}}}\in {\mathbb{C}^{N_{\rm t} \times 1}}$ and ${{\bf{u}}_{{\it m}}}\in {\mathbb{C}^{N_{\rm t} \times 1}}$ represent the beamformers of the common information $s_{\rm i,c}$ and the private information of the $m$-th communication user ${s_{ m}}$, $m \in \left\{ {1,2, \cdots ,M} \right\}$, respectively. 
As a result, the signal transmitted by the DFBS can be expressed as
\begin{equation}\label{signal}
{\bf{x}}\left( t \right) = {{\bf{u}}_{\rm{c}}}{s_{\rm{c}}}\left( t \right) + \sum\nolimits_{m = 1}^M {{{\bf{u}}_m}{s_m}\left( t \right)}  = {\bf{Us}}\left( t \right) \in {^{{N_{\rm{t}}} \times 1}},
\end{equation}
where ${{\bf{U}}} \in {\mathbb{C}^{N_{\rm t} \times \left( {M + 1} \right)}}$ is the transmit BF matrix and ${{\bf{s}}} \in {\mathbb{C}^{\left( {M + 1} \right) \times 1}}$ is the transmit signal vector. The receive signal at the $m$-th communication user can be obtained as
\begin{equation}\label{received_signal}
{y_m}\left( t \right) = {\bf{h}}_m^{\rm{H}}{{\bf{u}}_{\rm{c}}}{s_{\rm{c}}}\left( t \right) + {\bf{h}}_m^{\rm{H}}\sum\nolimits_{m = 1}^M {{{\bf{u}}_m}{s_m}\left( t \right)}  + {n_m}\left( t \right),
\end{equation}
where ${{\bf{h}}_m} \in {\mathbb{C}^{{N_{\rm{t}}} \times 1}}$ is the communication channel between the DFBS and the $m$-th user, and ${n_m}\left( t \right)$ is the \emph{additive white Gaussian noise} (AWGN) at the $m$-th user with expected value equal to 0 and variance being $\sigma^2$.
For the received signals,  the common information $s_{\rm c}$ is firstly decoded with the private information regarded as interference. Then, with the aid of the SIC, the common information is re-encoded, precoded, and eliminated from received signals \cite{6_ISAC,25_ISAC}. Afterwards, the private information of the $m$-th user $s_m$ is decoded with the other users’ private information deemed as the interference. Consequently, the \emph{signal-to-interference-plus-noise ratios} (SINRs) for the common and private streams of the $m$-th user can be respectively obtained as
\begin{equation}\label{common_SINR}
{\gamma _{\rm{c}}} = \frac{{\left\| {{\bf{h}}_m^{\rm{H}}{{\bf{u}}_{\rm{c}}}} \right\|_2^2}}{{\sum\nolimits_{j = 1}^M {\left\| {{\bf{h}}_m^{\rm{H}}{{\bf{u}}_j}} \right\|_2^2}  + {\sigma ^2}}},
\end{equation}
\begin{equation}\label{private_SINR}
{\gamma _m} = \frac{{\left\| {{\bf{h}}_m^{\rm{H}}{{\bf{u}}_m}} \right\|_2^2}}{{\sum\nolimits_{j \ne m} {\left\| {{\bf{h}}_m^{\rm{H}}{{\bf{u}}_j}} \right\|_2^2}  + {\sigma ^2}}}.
\end{equation}
Then, the corresponding common and private rates are respectively given by  ${R_{\rm{c}}} = {\log _2}\left( {1 + {\gamma _{\rm{c}}}} \right)$ and ${R_m} = {\log _2}\left( {1 + {\gamma _m}} \right)$.
The sum rate of the ISAC can be expressed as \cite{25_ISAC,RSMA_ISAC_Sensing_N}
\begin{equation}\label{sum_rate}
{R_{{\rm{sum}}}} = {R_{\rm{c}}} + \sum\nolimits_{m = 1}^M {{R_m}}    = \sum\nolimits_{m = 1}^M {\left( {{r_{{\rm{c}},m}} + {R_m}} \right)} ,
\end{equation}
where ${r_{{\rm{c}},m}}$ denotes the portion of the common rate allocated to the $m$-th user, and ${{\bf{r}}_{\rm{c}}} = {\left[ {{r_{{\rm{c,1}}}},{r_{{\rm{c,2}}}}, \cdots ,{r_{{\rm{c,}}M}}} \right]^{\rm{T}}} \in {\mathbb{C}^{M \times 1}}$.
We assume imperfect CSI for all the communication users. Thus, we have ${{\bf{h}}_{\rm{}}} = {\bf{h}}_{\rm{es}}^{{\rm{}}} + {\bf{h}}_{\rm{er}}^{{\rm{}}}$,
where $ {\bf{h}}_{\rm{es}}^{{\rm{}}}$ represents estimated CSI of the DFBS-user link, and ${\bf{h}}_{\rm{er}}$ is the error between the estimated CSI and the actual CSI of the DFBS-user link. The second norm of $\mathbf{h}_{\rm{es}}$ satisfies the following inequality: 
\begin{equation}\label{CSI_2}
	0 \le \left\| {{\bf{h}}_{\rm{er}}} \right\|_2 \le {{\rm e}_{\rm{h}}},
\end{equation} 
In combination with (\ref{CSI_2}), ${\left\| {{\bf{h}}_m^{\rm{H}}{{\bf{u}}_{\rm{c}}}} \right\|_2^2}$ in (\ref{common_SINR}) can be rewritten as
\begin{equation}
{\begin{array}{*{20}{l}}
		{\left\| {{\bf{h}}_m^{\rm{H}}{{\bf{u}}_{\rm{c}}}} \right\|_2^2}
	\end{array} = {\bf{u}}_{\rm{c}}^{\rm{H}}\left( {{{\bf{h}}_{{\rm{es}}}}{\bf{h}}_{{\rm{es}}}^{\rm{H}} + {{\bf{\Theta }}_{\rm{h}}}} \right){{\bf{u}}_{\rm{c}}},}
\end{equation} 
where channel error matrix ${{{\bf{\Theta }}_{\rm{h}}}}$ is satisfied with the trigonometric inequality constraint and compatibility, implying that
\begin{equation}
{{{\left\| {{{\bf{\Theta }}_{\rm{h}}}} \right\|}_2} \le {\rm{e}}_{\rm{h}}^2 + 2{{\rm{e}}_{\rm{h}}}{{\left\| {{{\bf{h}}_{{\rm{es}}}}} \right\|}_2} = {{\rm{e}}_{{\rm{h}},{\rm{max}}}}}.
\end{equation}
Similarly, for the private streams, by inserting equation (\ref{CSI_2}) into the molecule to the right of (\ref{private_SINR}), it holds that
\begin{equation}
\begin{array}{*{20}{l}}
	{\left\| {{\bf{h}}_m^{\rm{H}}{{\bf{u}}_m}} \right\|_2^2}
\end{array} = {\bf{u}}_m^{\rm{H}}\left( {{{\bf{h}}_{{\rm{es}}}}{\bf{h}}_{{\rm{es}}}^{\rm{H}} + {{\bf{\Theta }}_{\rm{h}}}} \right){{\bf{u}}_m}.
\end{equation}
Moreover, by defining ${{\bf{H}}_m} = {{\bf{h}}_m}{\bf{h}}_m^{\rm{H}}$, ${{\bf{U}}_{\rm{c}}} = {{\bf{u}}_{\rm{c}}}{\bf{u}}_{\rm{c}}^{\rm{H}}$, and ${{\bf{U}}_m} = {{\bf{u}}_m}{\bf{u}}_m^{\rm{H}}$, we get
\begin{equation}\label{ICSI_1}
\max {\mathop{\rm Tr}\nolimits} \left( {{{\bf{H}}_m}{{\bf{U}}_{\rm{c}}}} \right) = {\mathop{\rm Tr}\nolimits} \left[ {\left( {{{\bf{H}}_{{\rm{es}}}} + {{\rm{e}}_{{\rm{h,max}}}}{{\bf{I}}_{{N_{\rm{t}}} \times {N_{\rm{t}}}}}} \right){{\bf{U}}_{\rm{c}}}} \right],
\end{equation}
\begin{equation}\label{ICSI_2}
\min {\mathop{\rm Tr}\nolimits} \left( {{{\bf{H}}_m}{{\bf{U}}_{\rm{c}}}} \right) = {\mathop{\rm Tr}\nolimits} \left[ {\left( {{{\bf{H}}_{{\rm{es}}}} - {{\rm{e}}_{{\rm{h,max}}}}{{\bf{I}}_{{N_{\rm{t}}} \times {N_{\rm{t}}}}}} \right){{\bf{U}}_{\rm{c}}}} \right],
\end{equation}
\begin{equation}\label{ICSI_3}
\max {\mathop{\rm Tr}\nolimits} \left( {{{\bf{H}}_m}{{\bf{U}}_m}} \right) = {\mathop{\rm Tr}\nolimits} \left[ {\left( {{{\bf{H}}_{{\rm{es}}}} + {{\rm{e}}_{{\rm{h,max}}}}{{\bf{I}}_{{N_{\rm{t}}} \times {N_{\rm{t}}}}}} \right){{\bf{U}}_m}} \right],
\end{equation}
\begin{equation}\label{ICSI_4}
\min {\mathop{\rm Tr}\nolimits} \left( {{{\bf{H}}_m}{{\bf{U}}_m}} \right) = {\mathop{\rm Tr}\nolimits} \left[ {\left( {{{\bf{H}}_{{\rm{es}}}} - {{\rm{e}}_{{\rm{h,max}}}}{{\bf{I}}_{{N_{\rm{t}}} \times {N_{\rm{t}}}}}} \right){{\bf{U}}_m}} \right].
\end{equation}
To enhance the accuracy, reliability, and robustness of the ISAC system in resolving, tracking, and sensing multiple targets while striking a balance between resource usage and system complexity, the communication and sensing processes within one cycle of the ISAC system are discretized by employing multiple signal transmissions and radar pulses. This design can effectively address the complex requirements in multi-target sensing scenarios. Consequently, It is assumed that there exist $T$ transmission and radar pulse blocks in a coherent processing interval, and $t \in \left\{ {1,2, \cdots ,T} \right\}$. When $T$ is large enough, the difference between the sample covariance matrix of $\textbf{X}$ and the statistical covariance matrix ${\textbf{R}_{\bf{x}}}$ approaches $\textbf{0}$, i.e.
\begin{equation}\label{covariance}
{\textbf{R}_{\bf{x}}} = {T^{ - 1}}\sum\nolimits_{t = 1}^T {{\bf{x}}\left( t \right){\bf{x}}{{\left( t \right)}^{\rm{H}}}}  = {\bf{U}}{{\bf{U}}^{\rm{H}}}.
\end{equation}

\subsection{Sensing Signal Model}
The echo signal at the DFBS can be represented as
\begin{equation}\label{echo}
{{\bf{y}}_{{\rm{BS}}}} = \sum\nolimits_{k = 1}^K {\sqrt {{o_k}} {p_k}{{\bf{q}}_{\rm{r}}}\left( {{\varphi _k}} \right){\bf{q}}_{\rm{t}}^{\rm{H}}\left( {{\varphi _k}} \right)} {\bf{x}}\left( t \right) + {\bf{n}}\left( t \right),
\end{equation}
where $o_k$ denotes the weighting coefficient for the sensing power allocated to the $k$-th target, $p_k$ is the complex coefficient of the 
$k$-th target, of which the amplitude characterizes the round-trip path loss and is proportional to its \emph{radar cross section} (RCS) with $\mathbb{E}\left( {{{\left| {{p_k}} \right|}^2}} \right) = {{\rm{p}}_0}$, $\phi_k$ is the $k$-th target's angle of departure, ${{\bf{q}}_{\rm{t}}}\left( {{\phi _k}} \right) \in {\mathbb{C}^{{N_{\rm{t}}} \times 1}}$ and ${{\bf{q}}_{\rm{r}}}\left( {{\phi _k}} \right) \in {\mathbb{C}^{{N_{\rm{t}}} \times 1}}$ are the transmit and receive steering vectors with respect to $\phi_k$, and ${\bf{n}}\left( t \right) \in \mathbb{CN} \left( {{\bf{0}},{\sigma ^2}{{\bf{I}}_{{N_{\rm{r}}} \times 1}}} \right)$ represents  the noise at the DFBS. We model the clutter from the environment as complex random Gaussian noise with a mean being equal to 0 and a variance being equal to  $\sigma^2$. The \emph{signal-to-cluster-plus-noise ratio} (SCNR) of the $k$-th user's echo signal at the DFBS can be obtained as
\begin{equation}\label{gamma_echo}
\begin{array}{l} \displaystyle
	{\gamma _k} = \frac{{{o_k}{{\rm{p}}_0}{\rm{Tr}}\left[ {{{\bf{q}}_{\rm{r}}}\left( {{\varphi _k}} \right){\bf{q}}_{\rm{t}}^{\rm{H}}\left( {{\varphi _k}} \right){{\bf{q}}_{\rm{t}}}\left( {{\varphi _k}} \right){\bf{q}}_{\rm{r}}^{\rm{H}}\left( {{\varphi _k}} \right)} \right]}}{{{\sigma ^2}}} 
	= {o_k}{\gamma _{o,k}}, 
\end{array}
\end{equation}
where ${\bf{o}} = {\left\{ {{o_1},{o_2}, \cdots ,{o_K}} \right\}^{\rm{T}}} \in {\mathbb{C}^{K \times 1}}$ is the sensing power allocation vector.
Let ${{\bf{Y}}_{{\rm{BS}}}} = \left[ {{{\bf{y}}_{{\rm{BS}}}}\left( 1 \right),{{\bf{y}}_{{\rm{BS}}}}\left( 2 \right), \cdots ,{{\bf{y}}_{{\rm{BS}}}}\left( T \right)} \right] \in  {\mathbb{C}^{{N_{\rm{r}}} \times T}}$, ${{\bf{Q}}_{\rm{t}}} = \left[ {{{\bf{q}}_{\rm{t}}}\left( {{\varphi _1}} \right),{{\bf{q}}_{\rm{t}}}\left( {{\varphi _2}} \right), \cdots ,{{\bf{q}}_{\rm{t}}}\left( {{\varphi _K}} \right)} \right] \in {\mathbb{C}^{{N_{\rm{t}}} \times K}}$, ${{\bf{Q}}_{\rm{r}}} = \left[ {{{\bf{q}}_{\rm{r}}}\left( {{\varphi _1}} \right),{{\bf{q}}_{\rm{r}}}\left( {{\varphi _2}} \right), \cdots ,{{\bf{q}}_{\rm{r}}}\left( {{\varphi _K}} \right)} \right] \in  {\mathbb{C}^{{N_{\rm{r}}} \times K}}$, ${\bf{P}} = {\mathop{\rm diag}\nolimits} \left( {{p_1},{p_2}, \cdots ,{p_K}} \right) \in {\mathbb{C}^{K \times K}}$, and ${\bf{N}} = \left[ {{\bf{n}}\left( 1 \right),{\bf{n}}\left( 2 \right), \cdots ,{\bf{n}}\left( T \right)} \right] \in {\mathbb{C}^{{N_{\rm{r}}} \times T}}$, then (\ref{echo}) can be rewritten as
\begin{equation}\label{echo2}
{{\bf{Y}}_{{\rm{BS}}}} = {{\bf{Q}}_{\rm{r}}}{\bf{PQ}}_{\rm{t}}^{\rm{H}}{\bf{X}} + {\bf{N}}.
\end{equation}
We vectorize (\ref{echo2}) and formalize the associated definitions. Then, we obtain 
\begin{equation}\label{echo3}
\begin{array}{l}
	{\bf{\bar y}} 
	= {\mathop{\rm vec}\nolimits} \left( {{{\bf{Q}}_{\rm{r}}}{\bf{PQ}}_{\rm{t}}^{\rm{H}}{\bf{X}}} \right) + {\bf{\bar n}} = {\bf{\bar v}} + {\bf{\bar n}},
\end{array}
\end{equation}
where ${\bf{\bar y}} = {\mathop{\rm vec}\nolimits} \left( {\bf{Y}} \right) \in {\mathbb{C}^{T{N_{\rm{r}}} \times 1}}$ and ${\bf{\bar n}} = {\mathop{\rm vec}\nolimits} \left( N \right) \in {\mathbb{C}^{T{N_{\rm{r}}} \times 1}}$. ${\bf{\bar y}} \sim CN\left[ {\left( {{{\bf{X}}^{\rm{T}}} \otimes {{\bf{I}}_{{N_{\rm{r}}} \times {N_{\rm{r}}}}}} \right){\bf{\bar w}},{\sigma ^2}{{\bf{I}}_{{N_{\rm{r}}} \times {N_{\rm{r}}}}}} \right]$ represents a \emph{circularly symmetric complex Gaussian} (CSCG) random vector, where its statistical properties are fully characterized by a zero-mean condition and covariance matrix ${{\sigma ^2}{{\bf{I}}_{{N_{\rm{r}}} \times {N_{\rm{r}}}}}}$. Building upon the preceding framework, we postulate that the DFBS possesses prior information regarding the multiple targets, which is subsequently leveraged to estimate the unknown parameters ${p_k}$ and $\phi_k$. Given ${\bf{p}} = {\left[ {{p_1},{p_2}, \cdots ,{p_K}} \right]^{\rm{T}}}$, we hold ${\bf{p}} = {{\bf{p}}_{{\mathop{\rm Re}\nolimits} }} + {\rm{j}}{{\bf{p}}_{{\mathop{\rm Im}\nolimits} }}$, ${{\bf{p}}_{{\mathop{\rm Re}\nolimits} }} = {\left[ {{\mathop{\rm Re}\nolimits} \left( {{p_1}} \right),{\mathop{\rm Re}\nolimits} \left( {{p_2}} \right), \cdots ,{\mathop{\rm Re}\nolimits} \left( {{p_K}} \right)} \right]^{\rm{T}}}$, and ${{\bf{p}}_{{\mathop{\rm Im}\nolimits} }} = {\left[ {{\mathop{\rm Im}\nolimits} \left( {{p_1}} \right),{\mathop{\rm Im}\nolimits} \left( {{p_2}} \right), \cdots ,{\mathop{\rm Im}\nolimits} \left( {{p_K}} \right)} \right]^{\rm{T}}}$, respectively, implying the existence of $3K$ unknown parameters to be estimated, which can be compactly represented as an unknown parameter vector ${\bf{b}} = {\left[ {{{\bf{\phi }}^{\rm{T}}},{\bf{p}}_{{\mathop{\rm Re}\nolimits} }^{\rm{T}},{\bf{p}}_{{\mathop{\rm Im}\nolimits} }^{\rm{T}}} \right]^{\rm{T}}} \in {^{3K \times 1}}$.
Then, we get 
\begin{equation}\label{FIM}
{{\bf{F}}_{\bf{b}}}\left[ {i,j} \right] = \frac{2}{{{\sigma ^2}}}{\mathop{\rm Re}\nolimits} \left[ {\frac{{\partial {{{\bf{\bar v}}}^{\rm{H}}}}}{{\partial {\bf{b}}\left[ i \right]}} \cdot \frac{{\partial {\bf{\bar v}}}}{{\partial {\bf{b}}\left[ j \right]}}} \right].
\end{equation}
Following the approach that was described in \cite{MT_ISAC1}, we formalize the definitions 
\begin{equation}\label{deri_1}
{{\bf{\dot Q}}_{i}} = \left[ {\frac{{\partial {{\bf{q}}_{i}}\left( {{\phi _1}} \right)}}{{\partial {\phi _1}}},\frac{{\partial {{\bf{q}}_{i}}\left( {{\phi _2}} \right)}}{{\partial {\phi _2}}}, \cdots ,\frac{{\partial {{\bf{q}}_{i}}\left( {{\phi _K}} \right)}}{{\partial {\phi _K}}}} \right], i \in \left\{ {{\rm{t}},{\rm{r}}} \right\}.
\end{equation}
 The elements of the \emph{fisher information matrix} (FIM) can be expressed as
 \begin{equation}\label{F12}
 	\begin{array}{l}  \displaystyle
 		{{\bf{F}}_{12}} = \left( {T{\bf{\dot Q}}_{\rm{r}}^{\rm{T}}{{\bf{Q}}_{\rm{r}}}} \right) \odot \left( {{\bf{PQ}}_{\rm{t}}^{\rm{T}}{\bf{R}}_{\bf{X}}^{\rm{T}}{{\bf{Q}}_{\rm{t}}}} \right)\\  \displaystyle \qquad
 		+ \left( {T{\bf{Q}}_{\rm{r}}^{\rm{T}}{{\bf{Q}}_{\rm{r}}}} \right) \odot \left( {{\bf{P\dot Q}}_{\rm{t}}^{\rm{T}}{\bf{R}}_{\bf{X}}^{\rm{T}}{{{\bf{\dot Q}}}_{\rm{t}}}} \right),
 	\end{array}
 \end{equation}
 \begin{equation}\label{F22}
 	{{\bf{F}}_{22}} = \left( {T{\bf{Q}}_{\rm{r}}^{\rm{T}}{{\bf{Q}}_{\rm{r}}}} \right) \odot \left( {{\bf{Q}}_{\rm{t}}^{\rm{T}}{\bf{R}}_{\bf{X}}^{\rm{T}}{{\bf{Q}}_{\rm{t}}}} \right).
 \end{equation}
\begin{equation}\label{F11}
\begin{array}{l} \displaystyle
	{{\bf{F}}_{11}} = \left( {T{\bf{\dot Q}}_{\rm{r}}^{\rm{T}}{{{\bf{\dot Q}}}_{\rm{r}}}} \right) \odot \left( {{\bf{PQ}}_{\rm{t}}^{\rm{T}}{\bf{R}}_{\bf{X}}^{\rm{T}}{{\bf{Q}}_{\rm{t}}}{{\bf{P}}^{\rm{T}}}} \right)\\ \displaystyle \qquad
	+ \left( {T{\bf{\dot Q}}_{\rm{r}}^{\rm{T}}{{\bf{Q}}_{\rm{r}}}} \right) \odot \left( {{\bf{PQ}}_{\rm{t}}^{\rm{T}}{\bf{R}}_{\bf{X}}^{\rm{T}}{{{\bf{\dot Q}}}_{\rm{t}}}{{\bf{P}}^{\rm{T}}}} \right)\\ \displaystyle \qquad
	+ \left( {T{\bf{Q}}_{\rm{r}}^{\rm{T}}{{\bf{Q}}_{\rm{r}}}} \right) \odot \left( {{\bf{P\dot Q}}_{\rm{t}}^{\rm{T}}{\bf{R}}_{\bf{X}}^{\rm{T}}{{{\bf{\dot Q}}}_{\rm{t}}}{{\bf{P}}^{\rm{T}}}} \right)\\ \displaystyle \qquad
	+ \left( {T{\bf{Q}}_{\rm{r}}^{\rm{T}}{{{\bf{\dot Q}}}_{\rm{r}}}} \right) \odot \left( {{\bf{P\dot Q}}_{\rm{t}}^{\rm{T}}{\bf{R}}_{\bf{X}}^{\rm{T}}{{\bf{Q}}_{\rm{t}}}{{\bf{P}}^{\rm{T}}}} \right),
\end{array}
\end{equation}

Thus, the FIM can be written as
\begin{equation}\label{FIM2} \displaystyle
{{\bf{F}}_{\bf{b}}} = \frac{2}{{{\sigma ^2}}}\left[ {\begin{array}{*{20}{c}} \displaystyle
		{{\mathop{\rm Re}\nolimits} \left( {{{\bf{F}}_{11}}} \right)}&{{\mathop{\rm Re}\nolimits} \left( {{{\bf{F}}_{12}}} \right)}&{ - {\mathop{\rm Im}\nolimits} \left( {{{\bf{F}}_{12}}} \right)}\\ \displaystyle
		{{\mathop{\rm Re}\nolimits} \left( {{\bf{F}}_{12}^{\rm{T}}} \right)}&{{\mathop{\rm Re}\nolimits} \left( {{{\bf{F}}_{22}}} \right)}&{ - {\mathop{\rm Im}\nolimits} \left( {{{\bf{F}}_{22}}} \right)}\\ \displaystyle
		{ - {\mathop{\rm Im}\nolimits} \left( {{\bf{F}}_{12}^{\rm{T}}} \right)}&{ - {\mathop{\rm Im}\nolimits} \left( {{\bf{F}}_{22}^{\rm{T}}} \right)}&{{\mathop{\rm Re}\nolimits} \left( {{{\bf{F}}_{22}}} \right)}
\end{array}} \right].
\end{equation}
The CRB matrix is defined as
\begin{equation}
{{\bf{C}}_{{\rm{CRB}}}} = {\bf{F}}_{\bf{b}}^{ - 1}.
\end{equation}
We employ the trace of the CRB matrix $\textbf{C}_{\rm CRB}$ as a metric to quantify the estimation accuracy for the unknown parameter vector $\textbf{b}$. Specifically, $\textbf{C}_{\rm CRB}$ provides a lower bound on the covariance matrix of any unbiased estimator. The trace operation ${\mathop{\rm Tr}\nolimits} \left( {{{\bf{C}}_{{\rm{CRB}}}}} \right)$ thereby yields a scalar measure of the total estimation error variance across all parameters in $\textbf{b}$.


Finally, it is assumed that magnitudes of all channels considered follow Rician distribution. In other words, ${\bf{H}} = \sqrt \omega  {{\bf{H}}_{{\rm{LoS}}}} + \sqrt {1 - \omega } {{\bf{H}}_{{\rm{NLoS}}}}$, \cite{6_ISAC,MT_ISAC_3,rician}, where $\omega  \in \left[ {0,1} \right]$ is the weighted factor of the \emph{line-of-sight} (LOS) component.

\section{Optimization Formulation and Solution}
\subsection{Optimization Formulation}
The design of ISAC waveforms can be approached by investigating the trade-off between distinct communication and radar sensing metrics. We select the communication rate and the CRB as the respective performance measures for communication and sensing. To this end, we formulate an RSMA-assisted ISAC waveform optimization problem that minimizes the trace of the CRB matrix, while enhancing the minimum communication rate requirements. The optimization problem can be formally expressed as
\begin{subequations} \label{P1}
	\begin{align}
		&{\bf{P1}}:\mathop {\min }\limits_{{{\bf{U}}_{\rm{c}}},{{\bf{U}}_m},{{\bf{r}}_{\rm{c}}},{\bf{o}}} \left\{ {{{\bf{C}}_{{\rm{CRB}}}}, - \min \left( {{r_{{\rm{c}},m}} + {R_m}} \right)} \right\}\\
		&	{\rm{s}}{\rm{.t.}} \quad {R_{\rm{c}}} \ge \sum\nolimits_{m = 1}^M {{r_{{\rm{c}},m}}} ,\\ 
		&\quad\quad\; {R_{\rm{c}}} \ge {I_{\rm{c}}},\\ 
		&\quad\quad\;  {R_{m}} \ge {I_{\rm{p}}},\\
		&\quad\quad\; {\mathop{\rm Tr}\nolimits} \left( {{{\bf{R}}_{\bf{X}}}} \right) = \left\| {{{\bf{u}}_{\rm{c}}}} \right\|_2^2 + \sum\nolimits_{m = 1}^M {\left\| {{{\bf{u}}_m}} \right\|_2^2}  \le {P_{\max }},\\
		&  \quad\quad\; \left| {{\gamma _i} - {\gamma _j}} \right| \le \rho_0 ,i \ne j,i,j \in \left\{ {1,2, \cdots ,K} \right\},
	\end{align}	
\end{subequations}
where (\ref{P1}a) is the optimization objective; (\ref{P1}b) is the constraint on the common rate splitting; (\ref{P1}c) and (\ref{P1}d) are the constraints on the common and private rates; (\ref{P1}e) is the constraint on the transmit power of the DFBS; (\ref{P1}f) is the constraint on sensing power allocation, with $\rho_0$ denoting the tolerance threshold for the difference between two sensing SCNRs of echo signals from any two targets, respectively.

\subsection{Optimization Solution}
The joint optimization of sensing and communications in (\ref{P1}a) inherently introduces a performance tradeoff between these functionalities due to the finite transmit power at the DFBS, where enhancements in both sensing and communication critically depend on power allocation. To quantify the prioritization of sensing versus communications during optimization, we introduce two weighting parameters, $\lambda_1$ and $\lambda_2$, constrained by ${\lambda _1} + {\lambda _2} = 1$. A larger $\lambda_1$ (or $\lambda_2$) indicates that the DFBS allocates more power resources to sensing (or communications). Additionally, an auxiliary variable $a$ representing the upper bound of the joint sensing-communication performance, is introduced. Minimizing $a$ is equivalent to achieving the joint is equivalent to achieving the joint optimality of sensing and communications. Consequently, the original problem P1 can be reformulated as the following Pareto-optimal problem of sensing and communications:
\begin{subequations} \label{P1_supp}
	\begin{align}
		&{\bf{P2}}:\mathop {\min }\limits_{{{\bf{U}}_{\rm{c}}},{{\bf{U}}_m},{{\bf{r}}_{\rm{c}}},{\bf{o}}} a\\
		&	{\rm{s}}{\rm{.t.}} \quad {\lambda _1}{{\bf{C}}_{{\rm{CRB}}}} \le a,\\ 
		&   \quad \quad\;    - {\lambda _2}\min \left( {{r_{{\rm{c}},m}} + {R_m}} \right) \le a,\\
		& \quad \quad\; (\ref{P1} \rm b)- (\ref{P1}f),
	\end{align}	
\end{subequations}
where parameters $\lambda_1$ and $\lambda_2$ are weighted coefficients for CRB minimization and user rate maximization, respectively, ${\lambda _1} + {\lambda _2} = 1$. By adjusting $\lambda_1$ and $\lambda_2$, we obtain the complete set of Pareto-optimal ISAC transmit BF schemes that span the tradeoff between minimizing CRB and maximizing user rate.

\begin{theorem}\label{theorem_1}
	The optimization problem in (28) can be equivalently expressed as
	\begin{subequations} \label{P2}
		\begin{align}
			&{\bf{P3}}:\mathop {\min }\limits_{{{\bf{U}}_{\rm{c}}},{{\bf{U}}_m},{{\bf{r}}_{\rm{c}}},{\bf{o}},{\bf{c}},\wp } a \\
			&{\rm{s}}{\rm{.t.}} \quad {\lambda _1}\sum\nolimits_{i = 1}^{3K} {{v_i}}  \le a,\\
			&\quad\quad\; (\ref{P1_supp} \rm c)-(\ref{P1_supp} \rm d),  (\ref{P1_a3}),(\ref{P1_b5_1})-(\ref{P1_d5_1}), (\ref{gamma_4}) \\
			&\quad\quad\; {\mathop{\rm rank}\nolimits} \left( {{{\bf{U}}_{\rm{c}}}} \right) == {\mathop{\rm rank}\nolimits} \left( {{{\bf{U}}_m}} \right) = =1.
		\end{align}	
	\end{subequations}
\end{theorem}
\begin{proof}
The proof of Theorem \ref{theorem_1} is given in Appendix A. 
\end{proof}

Evidently, with the exception of constraint (\ref{P2}d), all of the remaining objective and constraints have been transformed into convex and tractable forms. Given the discontinuous nature of (\ref{P2}d), where penalty function methods demonstrate superior efficacy compared to the SDR, we further integrate (\ref{P2}d) and rewrite (\ref{P2}b) as
\begin{equation}\label{CRB_penalty}
\begin{array}{l}  \displaystyle
	{\lambda _1}\sum\nolimits_{i = 1}^{3K} {{v_i}}  \le a - {\Xi _{\rm{c}}}\left[ {{\rm{Tr}}\left( {{{\bf{U}}_{\rm{c}}}} \right) - {{\left( {{{\bf{u}}_{{\rm{c}},i}}{\bf{u}}_{{\rm{c}},i}^{\rm{H}}} \right)}^{\rm{T}}}{{\bf{U}}_{\rm{c}}}} \right]\\ \displaystyle
	- \sum\nolimits_{m = 1}^M {{\Xi _m}\left[ {{\rm{Tr}}\left( {{{\bf{U}}_m}} \right) - {{\left( {{{\bf{u}}_{m,i}}{\bf{u}}_{m,i}^{\rm{H}}} \right)}^{\rm{T}}}{{\bf{U}}_m}} \right]} 
\end{array},
\end{equation}
where ${\Xi _{\rm{c}}}$ and ${\Xi _m}$ are penalty factors of the penalty function in (\ref{P3_pro}a), $\textbf{u}_{{\rm c},i}$ and $\textbf{u}_{m,i}$ are the $i$-th iteration solutions of $\textbf{u}_{{\rm c}}$ and $\textbf{u}_{m}$. 
Then, we reformulate P3 as the following convex optimization problem:
\begin{subequations} \label{P3_pro}
	\begin{align}
		&{\bf{P4}}:
			\mathop {\min }\limits_{{{\bf{U}}_{\rm{c}}},{{\bf{U}}_m},{{\bf{r}}_{\rm{c}}},{\bf{o}},{\bf{c}},\wp }  a \\
		&{\rm{s}}{\rm{.t.}} \quad\; (\ref{P2} \rm c), \; (\ref{CRB_penalty}).
	\end{align}	
\end{subequations}
Hence, the convex targeted optimization P4 can be solved effectively through the successive iteration algorithm as shown in Algorithm 1 and related convex optimization tools. 

\subsection{Computation Complexity Analysis}
The computational complexity mainly results from the number of optimization variables and the number and the size of \emph{linear matrix inequality} (LMI) constraints in P4. For the targeted optimization P4, there exist $\left( {2N_{\rm{t}}^2 + N_{\rm{r}}^2} \right)$ original variables, $\left( {6M + 3K + 1} \right)$ slack variables, 1 $N_{\rm r}$-size LMI constraints, 2 $N_{\rm t}$-size LMI constraints, and $\left( {6M + 2K + 1} \right)$ 1-size LMI constraints. Therefore, the computational complexity of Algorithm 1 can be obtained as
\begin{equation}
\mathcal{O}\left\{T_{\rm iter} {\sqrt {{A_1}} {A_2}\left[ \begin{array}{l}
		2N_{\rm{t}}^2\left( {{N_{\rm{t}}} + {A_2}} \right) + N_{\rm{r}}^2\left( {{N_{\rm{r}}} + {A_2}} \right)\\ \displaystyle
		+ \left( {6M + 2K + 1} \right)\left( {1 + {A_2}} \right)
		 + A_2^2
	\end{array} \right]} \right\},
\end{equation}
where ${A_1} = 2{N_{\rm{t}}} + {N_{\rm{r}}} + 6M + 2K + 1$, ${A_2} = 2N_{\rm{t}}^2 + N_{\rm{r}}^2 + 6M + 3K + 1$ and $T_{\rm iter}$ represents the numbers of iterations for Algorithm \ref{alg:reference label}.

\begin{algorithm}[t]
	\caption{Fair-sensing-enabled RSMA-based ISAC iterative optimization algorithm.}
	\label{alg:reference label}
	
	{\algorithmicrequire} Channel vector ${\bf{h}}_m$, maximum transmit power $P_{\rm max}$.
	
	{\algorithmicensure} Pareto optimal common rate splitting vector ${{\bf{r}}_{\rm{c}}}$, BF vectors ${{\bf{u}}_{\rm{c}}}$ and ${{\bf{u}}_m}$, the sensing power allocation vector $\bf{o}$.
	
	\begin{algorithmic}[1] 
		\State Set ${\bf{u}}_{\rm{c},0}$ and ${\bf{u}}_{m,0}$ with the zero-forcing precoding scheme;
		\State Set positive thresholds $I_{\rm c}$, $I_{\rm p}$, $\rho_0$, $\rho_1$, $\rho_2$, and $\varepsilon=0$;
		\State Set penalty factors ${\Xi _{\rm{c}}}$ and ${\Xi _m}$;
		\State Initialize ${{\bf{R}}_{{\rm{c,}}\left( 0 \right)}}$, ${{\bf{O}}_{\left( 0 \right)}}$, ${{\bf{U}}_{{\rm{c,}}\left( 0 \right)}}$, and ${{\bf{U}}_{{{m,}}\left( 0 \right)}}$ to satisfy the constraints given in (\ref{P3_pro}b);
		\While {$\left| {{\mathop{\rm Tr}\nolimits} \left( {{{\bf{R}}_{{\rm{c,}}\left( \varepsilon  \right)}}} \right) - {\mathop{\rm Tr}\nolimits} \left( {{{\bf{R}}_{{\rm{c,}}\left( {\varepsilon  + 1} \right)}}} \right)} \right| \ge {\rho _1}$ or $\left| {{\mathop{\rm Tr}\nolimits} \left( {{{\bf{O}}_{\left( \varepsilon  \right)}}} \right) - {\mathop{\rm Tr}\nolimits} \left( {{{\bf{O}}_{\left( {\varepsilon  + 1} \right)}}} \right)} \right| \ge {\rho _1}$ or $\left| {{\mathop{\rm Tr}\nolimits} \left( {{{\bf{U}}_{{\rm{c}},\left( \varepsilon  \right)}}} \right) - {\mathop{\rm Tr}\nolimits} \left( {{{\bf{U}}_{{\rm{c}},\left( {\varepsilon  + 1} \right)}}} \right)} \right| \ge {\rho _1}$ or $\left| {{\mathop{\rm Tr}\nolimits} \left( {{{\bf{U}}_{m,\left( \varepsilon  \right)}}} \right) - {\mathop{\rm Tr}\nolimits} \left( {{{\bf{U}}_{m,\left( {\varepsilon  + 1} \right)}}} \right)} \right| \ge {\rho _1}$}
		\State Set $s=0$ and initialize auxiliary vector ${{\bf{c}}_{\left( s \right)}}$;
		\Repeat 
		\State $s=s+1$;
		\State Solve the targeted optimization in (\ref{P3_pro});
	   \State  Update related variables based on (\ref{P1_b5_1}) - (\ref{P1_d5_1});
		\State  Obtain $\gamma_i$, $i \in \left\{ {1,2, \cdots ,K} \right\}$;
		\Until {$\left| {{\gamma _i} - {\gamma _j}} \right| \le {\rho _0},i \ne j,i,j \in \left\{ {1,2, \cdots ,K} \right\}$}
		\State Calculate solutions ${{\bf{R}}_{{\rm{c,}}}}$, ${\bf{O}}$, ${{{\bf{U}}_{{\rm{c}}}}}$, and ${{{\bf{U}}_{{m}}}}$;
		\State Set ${{\bf{R}}_{{\rm{c,}}\left( {\varepsilon  + 1} \right)}} = {{\bf{R}}_{\rm{c}}}$ and ${{\bf{O}}_{\left( {\varepsilon  + 1} \right)}} = {\bf{O}}$;
		\State Set ${{\bf{U}}_{{\rm{c}},\left( {\varepsilon  + 1} \right)}} = {{\bf{U}}_{\rm{c}}}$ and ${{\bf{U}}_{m,\left( {\varepsilon  + 1} \right)}} = {{\bf{U}}_m}$;
		\If{${{\bf{R}}_{{\rm{c,}}\left( {\varepsilon  + 1} \right)}} = {{\bf{R}}_{{\rm{c,}}\left( \varepsilon  \right)}}$, ${{\bf{O}}_{\left( {\varepsilon  + 1} \right)}} = {{\bf{O}}_{\left( \varepsilon  \right)}}$, ${{\bf{U}}_{{\rm{c}},\left( \varepsilon  \right)}}$, and ${{\bf{U}}_{m,\left( {\varepsilon  + 1} \right)}} = {{\bf{U}}_{m,\left( \varepsilon  \right)}}$}
		\State Set ${\Xi _{\rm{c}}} = \left( {1 + {\rho _2}} \right){\Xi _{\rm{c}}}$ and ${\Xi _m}= \left( {1 + {\rho _2}} \right){\Xi _m}$;
		\EndIf
		\EndWhile
		\State Employ the singular value decomposition to solutions and then obtain ${{\bf{r}}_{{\rm{c}}}}$, $\bf{o}$, ${{\bf{u}}_{{\rm{c}}}}$, and ${{\bf{u}}_{{m}}}$;
	\end{algorithmic}
\end{algorithm} 
\renewcommand{\algorithmicrequire}{ \textbf{Input:}}     
\renewcommand{\algorithmicensure}{ \textbf{Output:}}    

\section{Numerical Results and Discussions}

\begin{figure}[htbp]
	\centering
	\includegraphics[width=8.75cm,height=7.25cm]{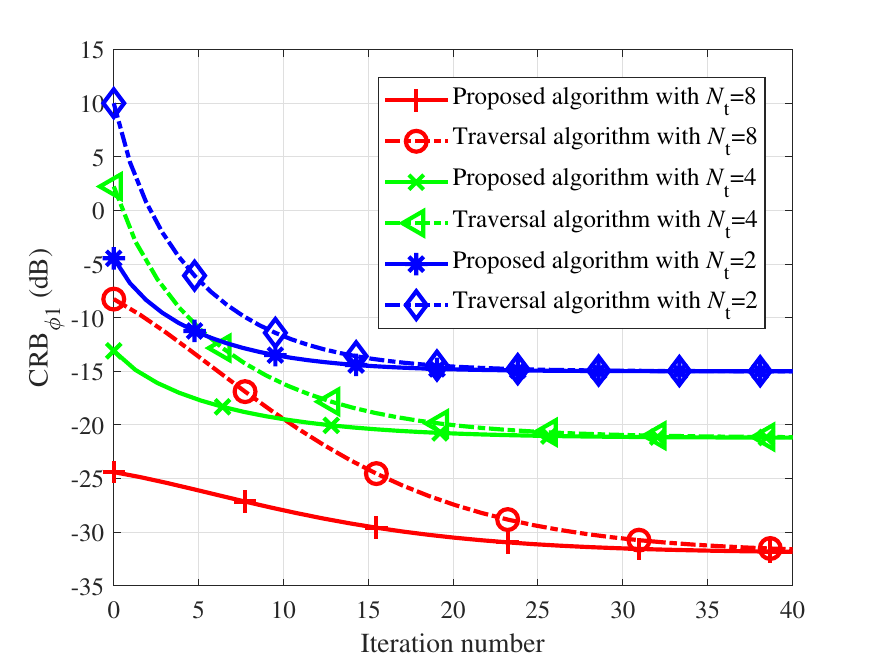}
	\caption{Convergence performance under different algorithms.}\label{figure2} 
\end{figure}

\begin{figure}[htbp]
	\centering
	\includegraphics[width=8.75cm,height=7.25cm]{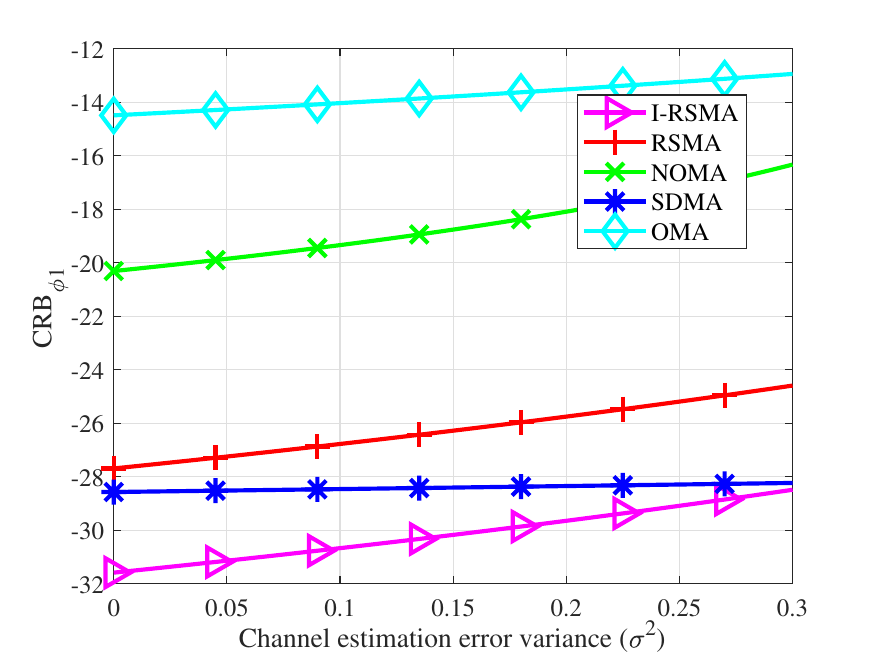}
	\caption{CRB versus CSI uncertainty for different MA schemes.}\label{figure3} 
\end{figure}

\begin{figure*}[htbp]
	\centering
	\subfigure[$\rm CRB_{\phi 1}$ versus $P_{\rm T}$ for different sensing and MA schemes.]{
		\label{figure4a}
		\includegraphics[width=8.75cm,height=7.25cm]{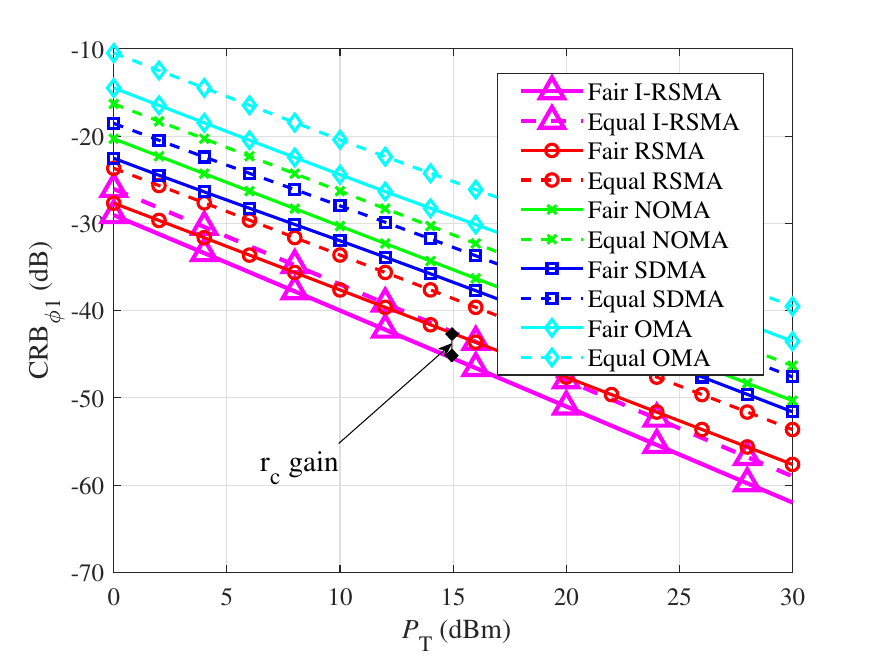}}
	\hspace{0.5mm}
	\subfigure[$\rm CRB_{\phi 2}$ versus $P_{\rm T}$ for different sensing and MA schemes.]{
		\label{figure4b} 
		\includegraphics[width=8.75cm,height=7.25cm]{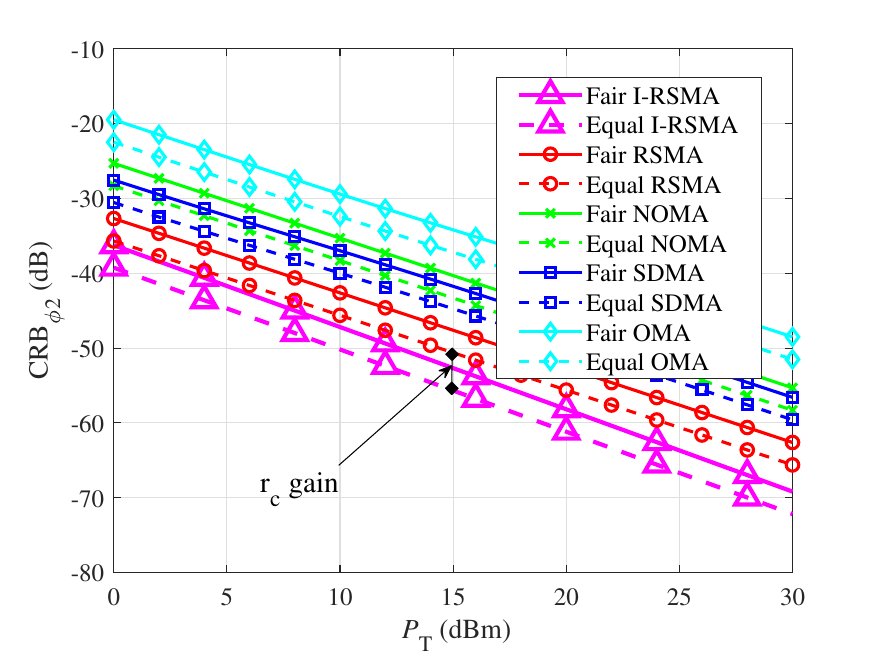}}
	\hspace{0.5mm}
	\subfigure[$\rm CRB_{\rm p 1}$ versus $P_{\rm T}$ for different sensing and MA schemes.]{
		\label{figure4c} 
		\includegraphics[width=8.75cm,height=7.25cm]{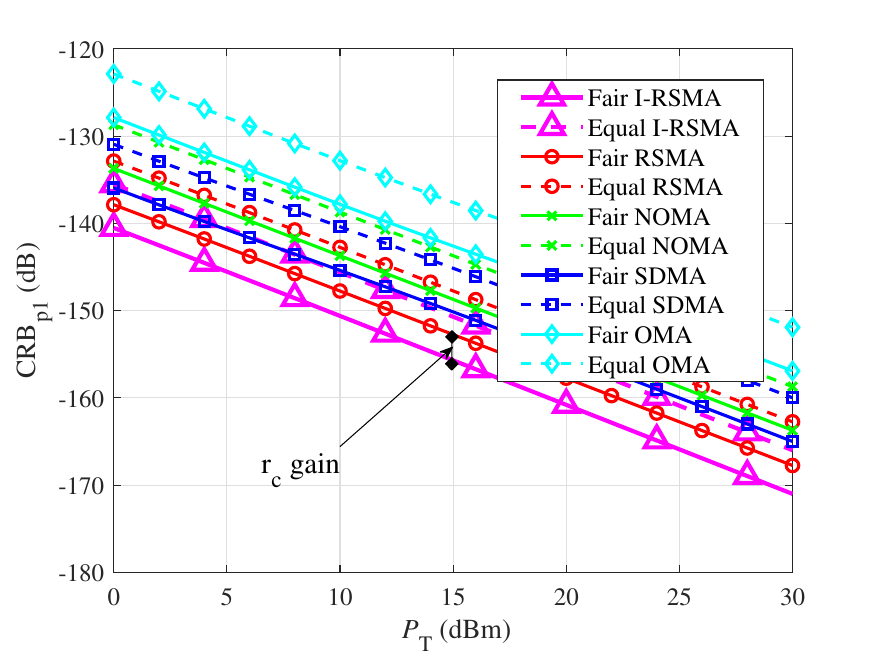}}
	\hspace{0.5mm}
	\subfigure[$\rm CRB_{\rm p 2}$ versus $P_{\rm T}$ for different sensing and MA schemes.]{
		\label{figure4d} 
		\includegraphics[width=8.75cm,height=7.25cm]{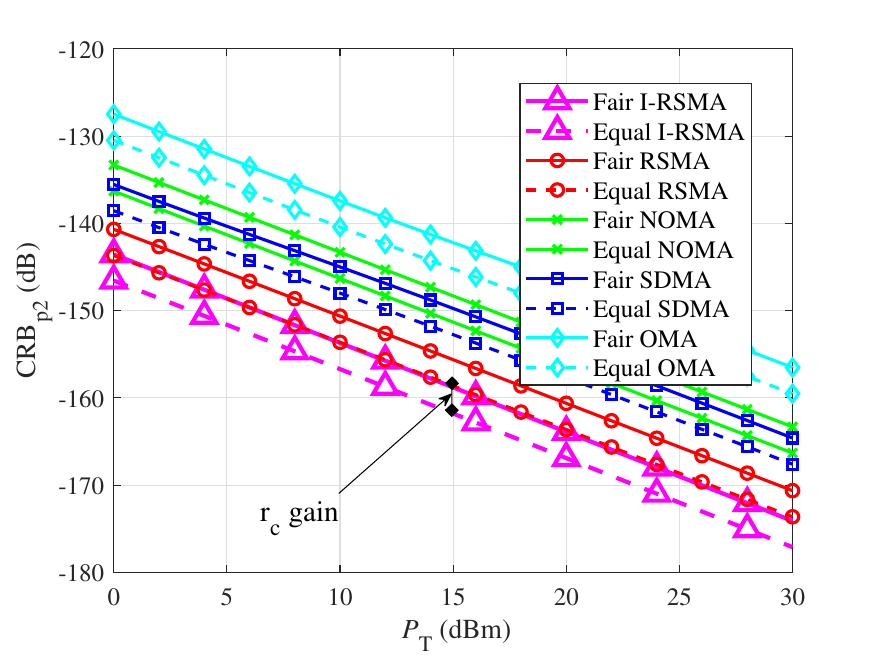}}
	\hspace{0.5mm}
	\caption{CRB versus transmit power for different sensing and MA schemes.}
	\label{figure4_dynamic} 
\end{figure*}

\begin{figure*}[htbp]
	\centering
	\subfigure[Tradeoff between $\rm CRB_{\phi 1}$ and $R_0$ for different sensing and MA schemes.]{
		\label{figure5a}
		\includegraphics[width=8.75cm,height=7.25cm]{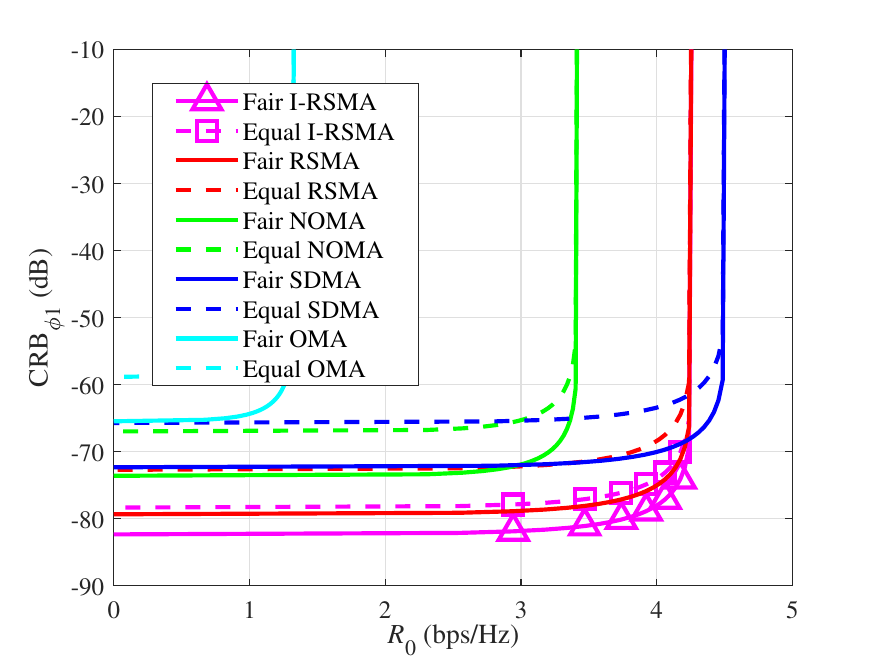}}
	\hspace{0.5mm}
	\subfigure[Tradeoff between $\rm CRB_{\phi 2}$ and $R_0$ for different sensing and MA schemes.]{
		\label{figure5b} 
		\includegraphics[width=8.75cm,height=7.25cm]{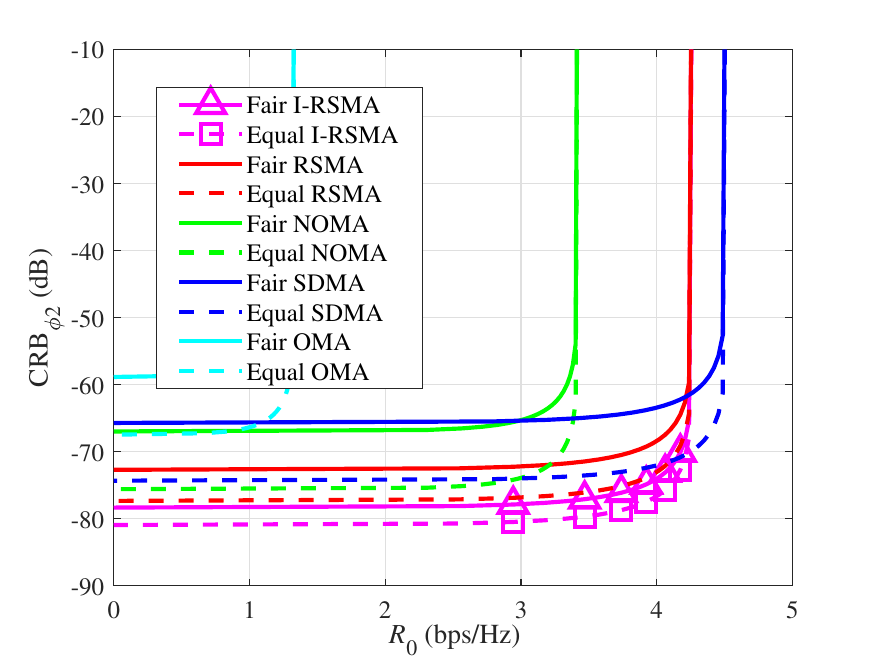}}
	\hspace{0.5mm}
	\subfigure[Tradeoff between $\rm CRB_{\rm p 1}$ and $R_0$ for different sensing and MA schemes.]{
		\label{figure5c}
		\includegraphics[width=8.75cm,height=7.25cm]{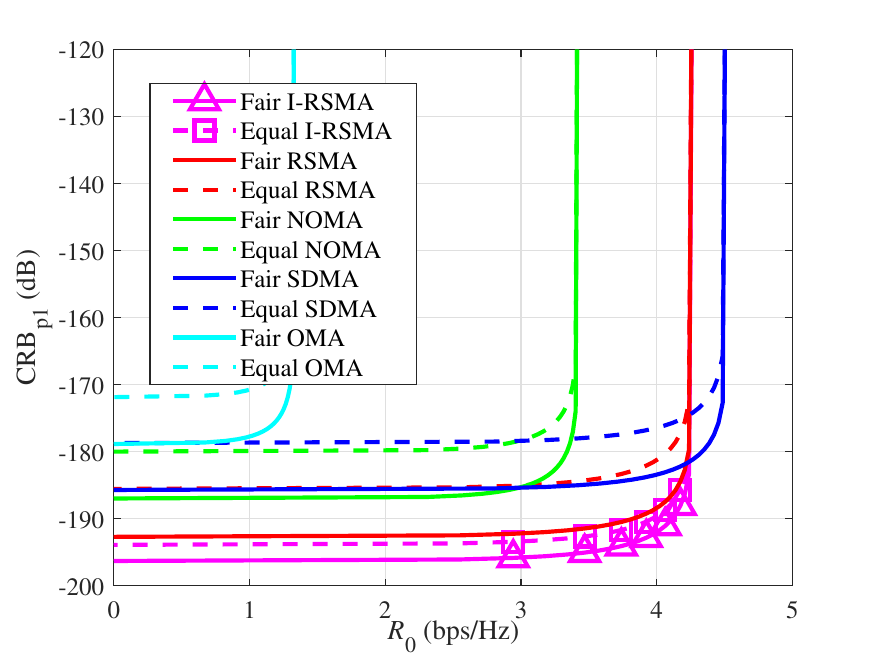}}
	\hspace{0.5mm}
	\subfigure[Tradeoff between $\rm CRB_{\rm p 2}$ and $R_0$ for different sensing and MA schemes.]{
		\label{figure5d}
		\includegraphics[width=8.75cm,height=7.25cm]{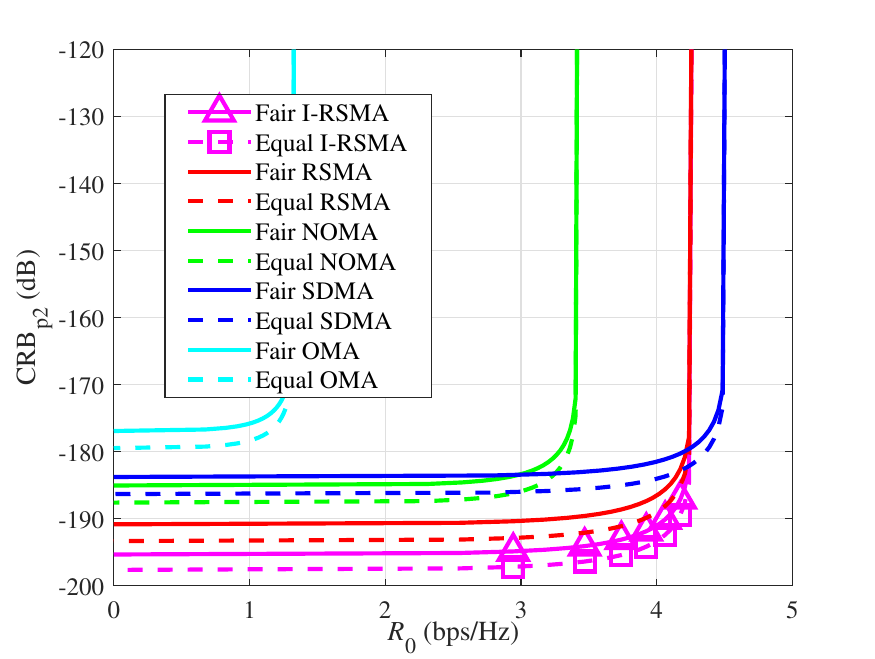}}
	\hspace{0.5mm}
	\caption{Pareto boundary on sensing and communications for different sensing and MA schemes.}
	\label{figure5_trade} 
\end{figure*}

\begin{figure*}[htbp]
	\centering
	\subfigure[3D fairness-aware RSMA-based beampattern.]{
		\label{figure6a}
		\includegraphics[width=8.75cm,height=7.25cm]{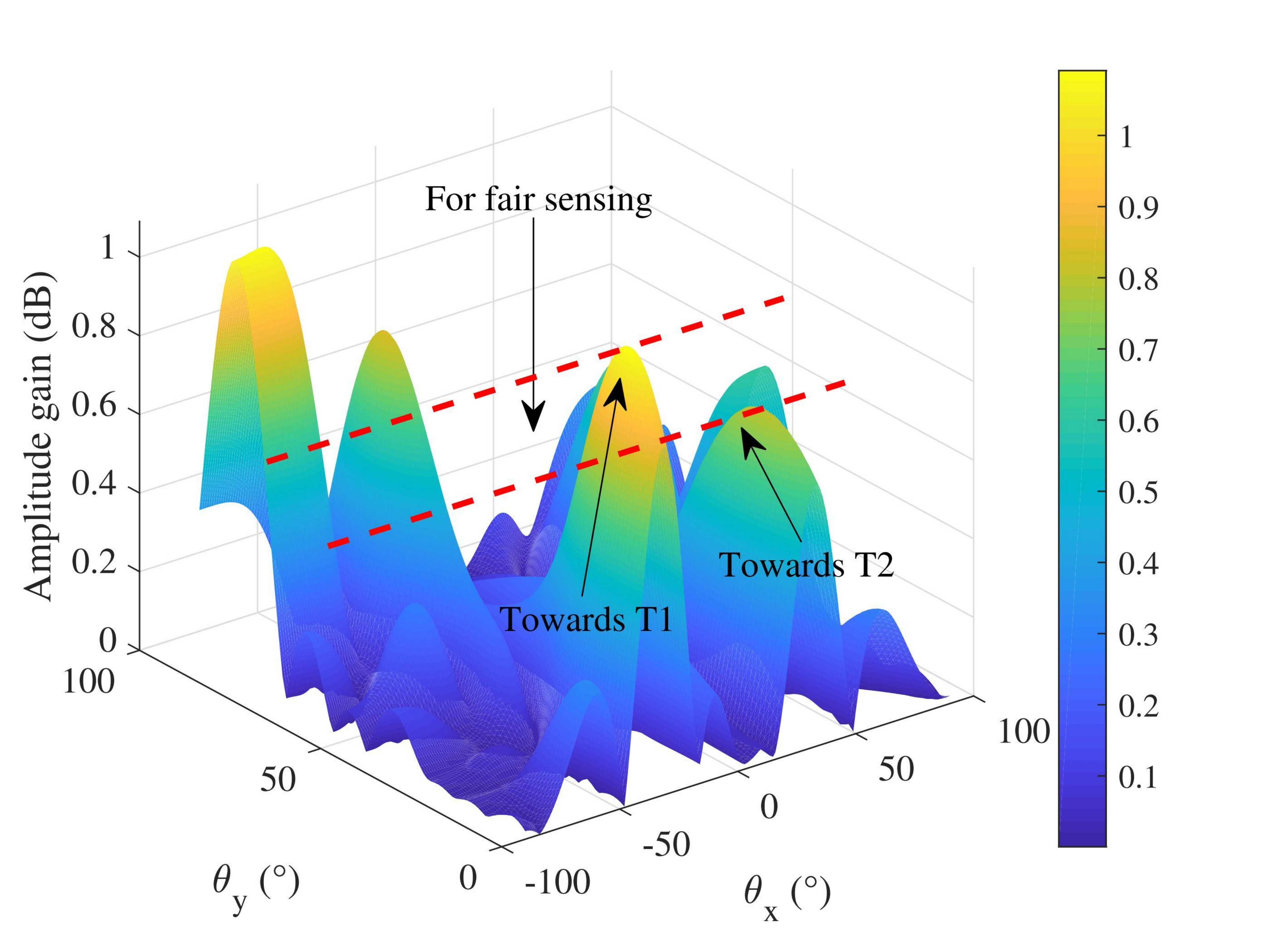}}
	\hspace{0.5mm}
	\subfigure[2D fairness-aware RSMA-based beampattern.]{
		\label{figure6b}
		\includegraphics[width=8.75cm,height=7.25cm]{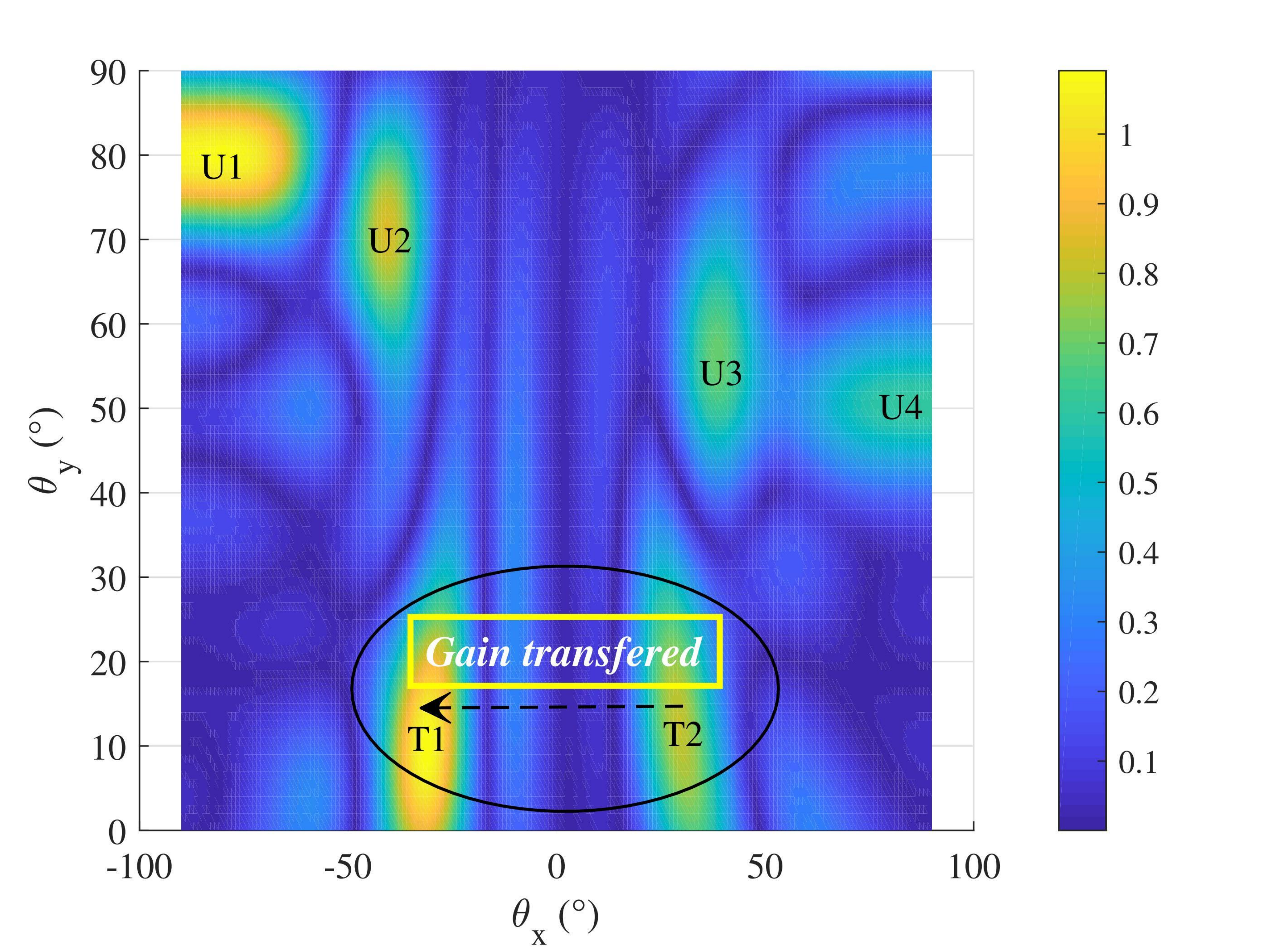}}
	\hspace{0.5mm}
	\subfigure[3D equality-aware RSMA-based beampattern.]{
		\label{figure6c} 
		\includegraphics[width=8.75cm,height=7.25cm]{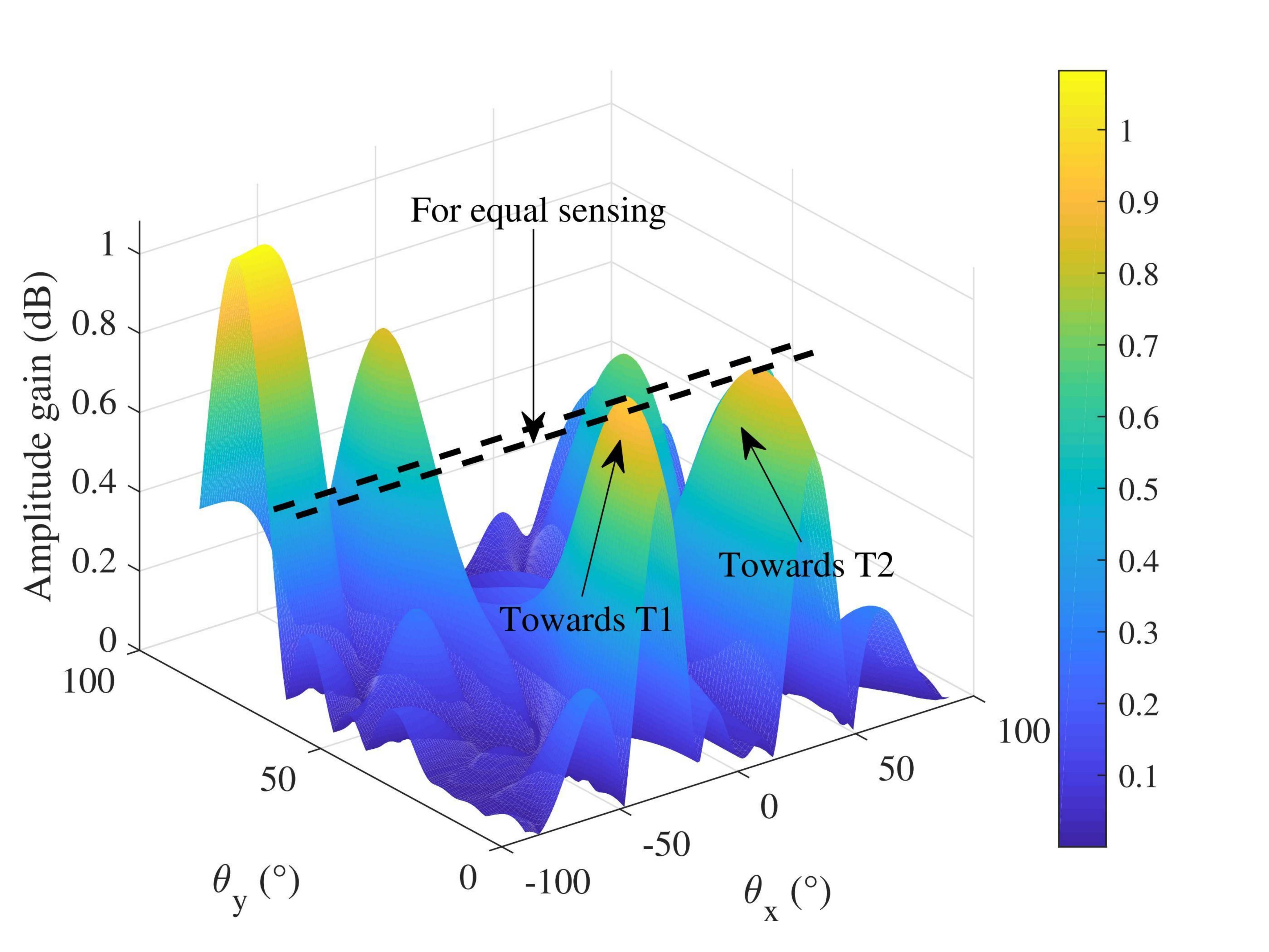}}
	\hspace{0.5mm}
		\subfigure[2D equality-aware RSMA-based beampattern.]{
		\label{figure6d} 
		\includegraphics[width=8.75cm,height=7.25cm]{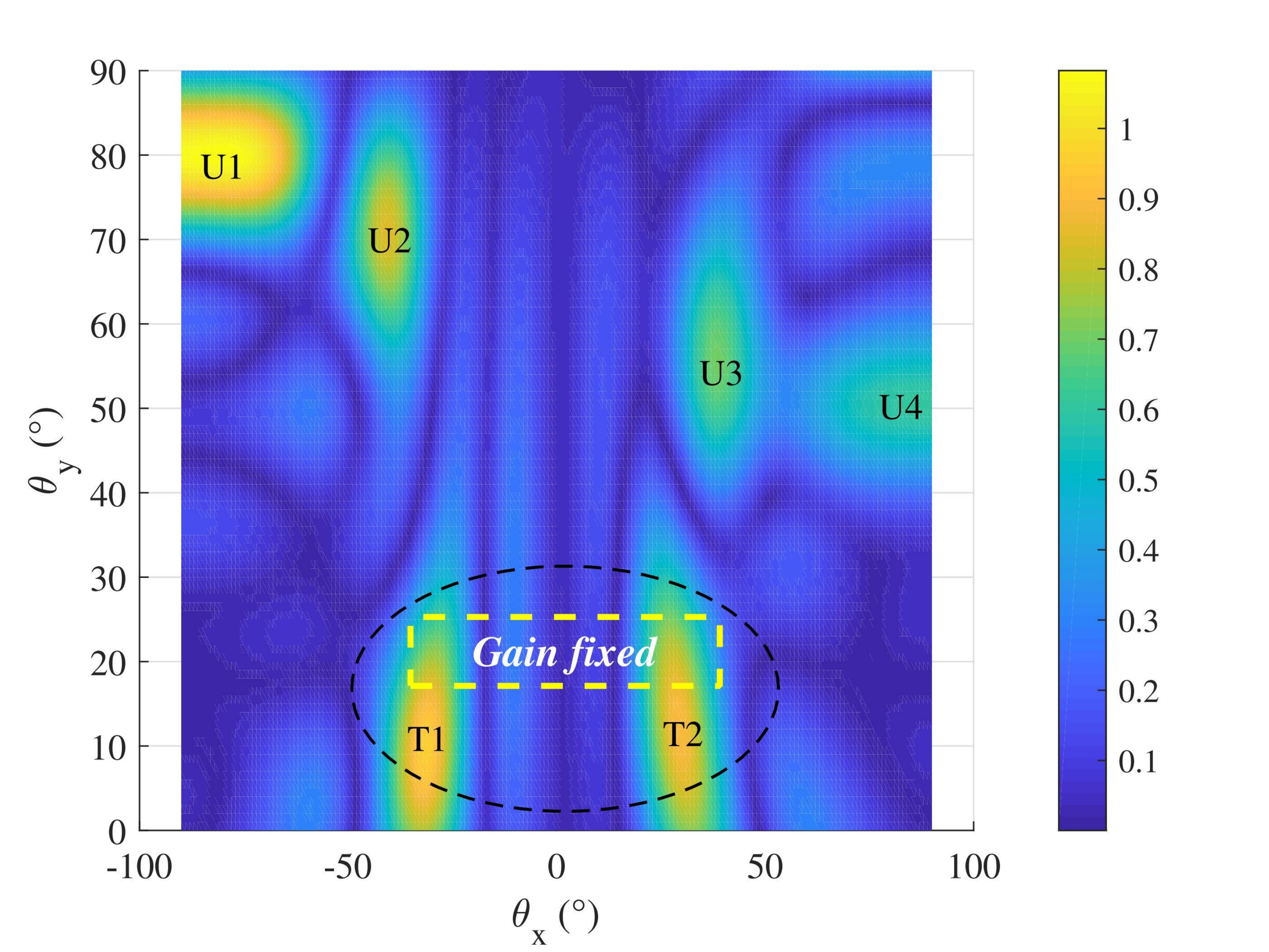}}
	\hspace{0.5mm}
	\caption{Comparison on fairness-aware and equality-aware beampatterns based on RSMA.}
	\label{figure6_trade} 
\end{figure*}

\begin{figure*}[htbp]
	\centering
	\subfigure[$\rm CRB_{\phi 1}$ and $\rm CRB_{\phi 2}$ by fairness-aware and equality-aware BF schemes.]{
		\label{figure7a}
		\includegraphics[width=8.75cm,height=5cm]{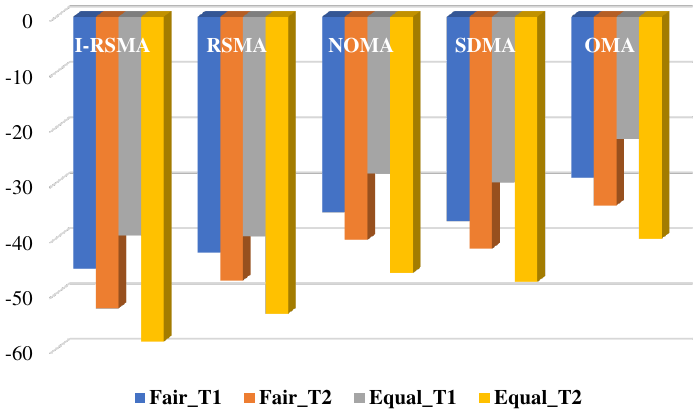}}
	\hspace{0.5mm}
	\subfigure[$\rm CRB_{\rm p 1}$ and $\rm CRB_{\rm p 2}$ by fairness-aware and equality-aware BF schemes.]{
		\label{figure7b}
		\includegraphics[width=8.75cm,height=5.1345cm]{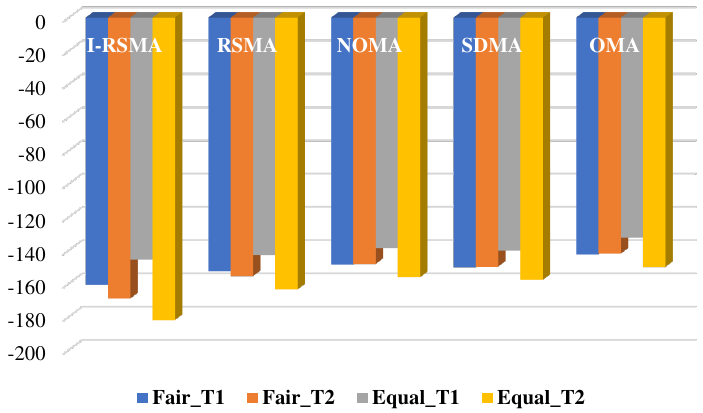}}
	\hspace{0.5mm}
	\caption{Comparison on CRBs achieved by fairness-aware and equality-aware BF schemes.}
	\label{figure7} 
\end{figure*}

In this section, Monte Carlo simulations are conducted to validate the effectiveness of the proposed methodology. The main simulation parameters are set as:  The bandwidth $f_{\rm m}$ = 80 MHz, the carrier frequency $f_{\rm c}$= 28 GHz, numbers of transmit and receive antennas $N_{\rm t}$=8, $N_{\rm r}$=8, number of targets and $K$=2, $M$=4, thresholds of common and private stream rates $I_{\rm c}$=1 bit/s/Hz, $I_{\rm p}$=1 bit/s/Hz, noise power at DFBS and users $\sigma^2$=0 dBm, maximum transmit power $P_{\rm max}$=40 dBm, iterative error tolerance $\rho_1$= $10^{-4}$, update step size of penalty factors $\rho_2$=0.3, weighted factor of the LoS component $\omega$= 0.5, and number of transmission and radar pulse blocks $T$=180, respectively \cite{RSMA_ISAC_Sensing_N,MT_ISAC_3}.
First, we demonstrate the convergence behavior of the Algorithm \ref{alg:reference label}. Next, we analyze the impact of varying CSI uncertainty parameters on sensing performance. Additionally, we illustrate the performance improvements in both sensing and communication functionalities achieved by the proposed ISAC framework. Finally, building upon Algorithm \ref{alg:reference label}, we present an RSMA-based sensing beampattern to confirm the effectiveness of the designed scheme.

In the subsequent analysis, we adopt a dual-target ISAC scenario as a representative case study for the multi-target sensing in ISAC. The proposed framework and analytical results remain applicable to systems with an arbitrary number of sensing targets. For simplicity, let $\rm CRB_{{\phi }\it i}$ and $\rm CRB_{{p }\it i}$ denote the angle and complex-parameter CRBs for the $i$-th target, $i \in \left\{ {1,2} \right\}$. T1 and T2 stand for Target 1 and Target 2. In the preliminary sensing phase, the DFBS exhibits suboptimal sensing accuracy for T1, while demonstrating superior accuracy for T2. $R_0$ represents the minimum communication rate requirement for each user. \emph{Improved RSMA} (I-RSMA) refers to the RSMA with common rate splitting optimization.

The convergence performance of the proposed sensing enhancement scheme for the multi-target multi-user ISAC system is illustrated by the solid curve in Figure \ref{figure2}, with the traversal search scheme shown as a dashed curve for comparison. Specifically, under configurations where the DFBS is equipped with $N_{\rm t} = 2,4,8$ transmit antennas, the proposed algorithm achieves convergence within 16, 20, and 30 iterations, respectively, demonstrating superior convergence efficiency and effectiveness. In contrast, the traversal scheme requires 23, 30, and 40 iterations under the same configurations. Taking $N_{\rm t} =8$ as an example, while the traversal scheme converges within 40 iterations, Algorithm \ref{alg:reference label} achieves convergence in only 30 iterations. This acceleration is accompanied by a significant improvement in sensing performance, where the $\rm CRB_{\phi 1}$ decreases from -25 dB to -32 dB. These enhancements stem from the proposed first-order Taylor series expansion and SCA, which effectively transform the original non-convex optimization problem into a convex formulation, thereby simplifying the solution to P3. Furthermore, the reduced iteration count suggests that the proposed balanced allocation-based initialization already resides near the neighborhood of the optimal solution. On the other hand, for a fixed number of iterations, increasing the number of transmit antennas progressively improves the quality of echo signals received at the DFBS. This improvement arises because a larger antenna array provides higher spatial degrees of freedom, enabling more precise and effective BF. However, the increased computational complexity associated with larger $N_{\rm t}$ manifests as a gradual rise in the required iterations for convergence.

Figure \ref{figure3} displays the significant divergence in $\rm CRB_{\phi 1}$ across MA schemes for different channel state uncertainty parameters. RSMA exhibits the strongest robustness, leveraging its rate-splitting mechanism to suppress error impacts quadratically (CRB increases by only 3.50 dB at $\sigma^2 = 0.30$). SDMA mitigates error risks through spatial multiplexing, showing the slowest CRB growth, making it suitable for large-scale antenna systems. OMA displays linear sensitivity to channel uncertainty, with CRB rising by 0.50 dB per 0.10  increase in $\sigma^2$, due to its orthogonal resource isolation. Conversely, NOMA suffers from error amplification caused by superposition coding, leading to severe CRB degradation when ${\sigma ^2} > 0.15$. Besides, I-RSMA scheme has an additional 3.89 dB sensing performance gain compared to RSMA.
For engineering practice, RSMA should be prioritized in high-error scenarios (${\sigma ^2} > 0.15$), while SDMA’s spatial gains can be leveraged in low-to-moderate uncertainty regimes. NOMA requires stringent channel estimation accuracy (${\sigma ^2} < 0.10$). These findings provide a theoretical foundation for the adaptive selection of MA schemes in dynamic channel environments.

Figure \ref{figure4_dynamic} comprehensively demonstrates the impact of different MA schemes and transmit power levels on the sensing CRB.
Figure \ref{figure4a} provides the variation trend of the $\rm CRB_{\phi 1}$ under four MA schemes RSMA, NOMA, SDMA, and OMA, as the transmit power $P_{\rm T}$ increases from 0 dBm to 30 dBm. $\rm CRB_{\phi 1}$ of RSMA decreases from -23.69 dB to -53.62 dB, achieving a reduction of  29.93 dB, which demonstrates the best performance. SDMA exhibits a slightly inferior improvement, with its $\rm CRB_{\phi 1}$ decreasing from -18.55 dB to  -47.59 dB (reduction: 29.04 dB). NOMA shows a $\rm CRB_{\phi 1}$ reduction from -16.31 dB to -46.31 dB (reduction: 30.00 dB), while OMA performs the worst, decreasing from -10.49 dB to -39.51 dB (reduction: 29.02 dB). The superior estimation accuracy of RSMA and SDMA originates from their efficient beamforming designs: RSMA optimizes the directionality of sensing signals through the synergistic interaction of common and private streams, whereas SDMA enhances target resolution via spatial separation. The relatively weaker performance of NOMA stems from the additional interference introduced by power-domain multiplexing, which degrades sensing gains. OMA exhibits the poorest performance due to its orthogonal resource allocation, which disperses sensing power across multiple time-frequency resources, thereby reducing the power available for individual target sensing. Furthermore, the proposed I-RSMA scheme, which enhances sensing performance by optimizing the common rate allocation, achieves an additional 3.96 dB $\rm CRB_{\phi 1}$ improvement compared to conventional RSMA. This gain arises because the communication rate constraints require each user’s rate to exceed a threshold. By flexibly adjusting the common rate allocation, I-RSMA satisfies these constraints with less power, thereby reserving more power for target sensing. This mechanism validates the effectiveness of the proposed I-RSMA scheme in harmonizing sensing and communication functionalities.

Figure \ref{figure4b} depicts the variation of the $\rm CRB_{\phi 2}$ with $P_{\rm T}$, highlighting the impact of fair power allocation on sensing performance. Under the equal power allocation, RSMA, SDMA, NOMA, and OMA achieve $\rm CRB_{\phi 2}$ improvements of 29.93 dB, 29.04 dB, 30.00 dB, and 29.02 dB, respectively. However, the equal power allocation leads to significant disparities between $\rm CRB_{\phi 1}$ and $\rm CRB_{\phi 2}$, as the DFBS cannot adaptively allocate sensing power based on the distinct characteristics of each target to ensure fair multi-target sensing. RSMA and SDMA optimize the SCNR through BF, whereas NOMA is constrained by interference, and OMA suffers from resource dispersion, resulting in inferior performance. 
Under the fair power allocation, when the DFBS reduces the sensing power allocated to T2, the $\rm CRB_{\phi 2}$ for RSMA, SDMA, NOMA, and OMA increases by 3.93 dB, 4.04 dB, 4.50 dB, and 5.02 dB, respectively. RSMA and SDMA exhibit the best robustness, as their beamforming focus maintains partial sensing capability even at reduced power levels. In contrast, OMA’s performance degrades more severely due to its inherent resource fragmentation, making it highly sensitive to power reduction. Remarkably, the proposed I-RSMA scheme achieves an additional 5.16 dB sensing performance gain compared to conventional RSMA. This improvement stems from its adaptive common rate allocation, which optimally balances communication rate constraints and sensing power utilization, further validating the superiority of I-RSMA.

Figure \ref{figure4c} gives $\rm CRB_{p 1}$ as a function of $P_{\rm T}$, demonstrating the performance of the four schemes in sensing range-related parameters. Under the equal power allocation, RSMA, SDMA, NOMA, and OMA achieve $\rm CRB_{p 1}$ improvements of 29.81 dB, 29.15 dB, 29.02 dB, and 29.03 dB, respectively. When fair allocation is applied, where T1 receives additional sensing power, $\rm CRB_{p 1}$ is significantly reduced. Specifically, $\rm CRB_{p 1}$ decreases by 4.53 dB, 4.18 dB, 3.69 dB, and 3.43 dB for RSMA, SDMA, NOMA, and OMA, respectively. Notably, the proposed I-RSMA scheme attains a further 3.25 dB sensing performance gain compared to conventional RSMA. This enhancement arises from its fair power reallocation mechanism, which prioritizes critical sensing tasks while maintaining communication rate guarantees, thereby optimizing the ISAC tradeoff under the limited power.

Figure \ref{figure4d} presents the variation of $\rm CRB_{p 2}$ with$P_{\rm T}$, reflecting the performance of four schemes in sensing range-related parameters. Under the equal power allocation, RSMA, SDMA, NOMA, and OMA achieve $\rm CRB_{p 2}$ improvements of 29.92 dB, 29.03 dB, 30.02 dB, and 29.05 dB, respectively. However, when fair allocation is applied, where T2 sacrifices a portion of its sensing power, $\rm CRB_{p 2}$ increases by 3.72 dB, 3.81 dB, 4.38 dB, and 4.63 dB for RSMA, SDMA, NOMA, and OMA, respectively. Remarkably, the proposed I-RSMA scheme achieves a 3.21 dB sensing performance gain compared to conventional RSMA under the fair regime.

Figure \ref{figure5_trade} demonstrates the Pareto boundary on sensing and communication under different MA schemes.
Figure \ref{figure5a} illustrates the tradeoff between $\rm CRB_{\phi 1}$ and $R_0$, revealing the impact of power allocation on balancing sensing and communication performance. Notably, for RSMA, SDMA, NOMA, and OMA, the achieved $\rm CRB_{\phi 1}$ remains nearly constant when $R_0$ is below 4.28 bps/Hz, 4.49 bps/Hz, 3.41 bps/Hz, and 1.32 bps/Hz, named rate thresholds, respectively. This phenomenon arises because the rate constraint becomes inactive within these ranges, allowing the system to simultaneously minimize $\rm CRB_{\phi 1}$ and achieve satisfactory communication rates by exploiting the correlation between communication and sensing channels.
Under the equal power allocation, $\rm CRB_{\phi 1}$ at these rate thresholds is -61.72 dB (RSMA), -55.70 dB (SDMA), -56.97 dB (NOMA), and -50.68 dB (OMA). In contrast, the fair power allocation significantly enhances sensing performance, yielding $\rm CRB_{\phi 1}$ of -74.13 dB (RSMA), -67.75 dB (SDMA), -69.62 dB (NOMA), and -59.43 dB (OMA).
Figure \ref{figure5b} further analyzes this tradeoff for T2. Under the equal allocation, $\rm CRB_{\phi 1}$ at rate thresholds is -72.15 dB (RSMA), -67.32 dB (SDMA), -70.35 dB (NOMA), and -63.48 dB (OMA). With the fair allocation, these values shift to -65.73 dB (RSMA), -55.64 dB (SDMA), -58.73 dB (NOMA), and -50.68 dB (OMA).
Figure \ref{figure5c} shows the tradeoff between $\rm CRB_{p 1}$ and $R_0$. Under the equal power allocation, the $\rm CRB_{p 1}$ at rate thresholds are -177.36 dB (RSMA), -168.72 dB (SDMA), -171.73 dB (NOMA), and -164.92 dB (OMA). In contrast, the fair power allocation significantly enhances sensing accuracy, yielding $\rm CRB_{p 1}$ of -190.17 dB (RSMA), -180.55 dB (SDMA), -183.02 dB (NOMA), and -177.85 dB (OMA).
Figure \ref{figure5d} further analyzes this tradeoff for T2. Under the equal allocation, $\rm CRB_{p 1}$ at the rate thresholds are -190.76 dB (RSMA), -183.67 dB (SDMA), -185.46 dB (NOMA), and -178.25 dB (OMA). With the fair allocation, these values shift to -180.33 dB (RSMA), -170.72 dB (SDMA), -176.85 dB (NOMA), and -170.92 dB (OMA). Similarly, the proposed fair power allocation scheme achieves more balanced multi-target sensing performance, as evidenced by the reduced disparity between $\rm CRB_{p 1}$ and $\rm CRB_{p 2}$ compared to the equal one.

Figure \ref{figure6_trade} shows the beampattern of the first user’s private stream as an example to visually illustrates the difference between fairness-aware sensing and equality-aware sensing. In Figures \ref{figure6a} and \ref{figure6b}, the beampattern gain directed toward T1 is significantly higher than that for T2. This disparity arises from our fairness-aware BF design, which fairly prioritizes sensing resources based on initial estimation accuracy. Specifically, when the DFBS detects inferior estimation accuracy for T1 compared to T2, it allocates more sensing power to T1, and enhances its beampattern gain, while reducing power and beampattern gain to T2. This adaptive strategy ensures the balanced and fair sensing performance.
In stark contrast, Figures \ref{figure6c} and \ref{figure6d} demonstrate nearly identical beampattern gains for T1 and T2 under a equality-aware BF scheme, which fails to adjust beampattern gains based on target-specific characteristics, resulting in imbalanced sensing outcomes. These observations align consistently with the simulation results in Figures \ref{figure4a} – \ref{figure4d} and Figures \ref{figure5a} – \ref{figure5d}, validating robustness of the proposed fairness-aware BF scheme.

Figure \ref{figure7} illustrates the impacts of fairness-aware and equality-aware BF schemes on CRB. As evident in Figures \ref{figure7a} and \ref{figure7b}, under various MA scenarios, the disparities between $\rm CRB_{\phi 1}$ and $\rm CRB_{\phi 2}$ and between $\rm CRB_{p 1}$ and $\rm CRB_{p 2}$ achieved by fairness-aware BF is significantly smaller than those obtained via equality-aware BF. This demonstrates that the DFBS, through fairness-aware BF, achieves the  more balanced multi-target sensing performance with a limited power budget. These results align with those presented in Figures \ref{figure4a} – \ref{figure4d}, further validating the significance of the proposed fairness-aware BF in enabling fair sensing across heterogeneous targets.

\section{Conclusion}
This paper presented a novel fairness-aware RSMA-based sensing framework for ISAC systems to achieve enhanced sensing accuracy and communication quality. By jointly optimizing BF vectors, common rate splitting strategies, and sensing power allocation, the proposed scheme effectively optimized the CRB and communication rates. Analytical derivations of the CRB and communication rates revealed fundamental tradeoffs between sensing and communication functionalities. To address the non-convex nature of the optimization problem, the original formulation was transformed into a Pareto-optimal problem. A computationally efficient algorithm was developed by leveraging Taylor series expansion, SDR, SCA, and a penalty function approach. Simulations demonstrated that the proposed fairness-aware RSMA-based sensing scheme significantly outperformed NOMA, SDMA, and OMA in three critical aspects: 1) achieving up to 29.14\% lower CRB for target parameter estimation; 2) exhibiting superior flexibility in balancing sensing-communication performance tradeoffs; 3) enabling more accurate BF for simultaneous multi-target sensing and multi-user communication provisioning. These advancements highlighted the potential of RSMA as a key enabler for the next-generation ISAC systems requiring effectively fair resource coordination between sensing and communication. 
 
 \begin{appendices}
 \section{Proof of Theorem 1}
Given the non-convex nature of the optimization objective, we transform the left part of (\ref{P1_supp}b) into the following form
\begin{equation}\label{P1_a}
	{{\bf{C}}_{{\rm{CRB,}}i,i}} = {\bf{a}}_i^{\rm{T}}{\bf{F}}_{\bf{b}}^{ - 1}{{\bf{a}}_i} \, \text{with}\, i \in \left\{ {1,2, \cdots ,3K} \right\},
\end{equation}
where ${{\bf{a}}_i}$ denotes the $i$-th column vector of the identity matrix. Then, a set of auxiliary variables $\wp  = \left\{ {{v_1},{v_2}, \cdots ,{v_{3K}}} \right\}$ are introduced, of which each element satisfies 
\begin{equation}\label{P1_a2}
	{v_i} \ge {\bf{a}}_i^{\rm{T}}{\bf{F}}_{\bf{b}}^{ - 1}{{\bf{a}}_i}.
\end{equation}
Based on (\ref{P1_a2}) as well as Schur complement, there exists \cite{MT_ISAC1}
\begin{equation}\label{P1_a3}
	\left[ {\begin{array}{*{20}{c}}
			{{{\bf{F}}_{\bf{b}}}}&{{{\bf{a}}_i}}\\
			{{\bf{a}}_i^{\rm{T}}}&{{v_i}}
	\end{array}} \right] \ge \left( {\begin{array}{*{20}{c}}
			{{0_{1,1}}}& \ldots &{{0_{1,3K + 1}}}\\
			\vdots & \ddots & \vdots \\
			{{0_{3K + 1,1}}}& \cdots &{{0_{3K + 1,3K + 1}}}
	\end{array}} \right).
\end{equation}

Let ${{\bf{R}}_{\rm{c}}} = {{\bf{r}}_{\rm{c}}}{\bf{r}}_{\rm{c}}^{\rm{H}}$, ${{\bf{U}}_{\rm{c}}} = {{\bf{u}}_{\rm{c}}}{\bf{u}}{}_{\rm{c}}^{\rm{H}} \ge {\bf{0}}$, ${{\bf{U}}_m} = {{\bf{u}}_m}{\bf{u}}{}_m^{\rm{H}}\ge {\bf{0}}$, ${{\bf{H}}_i} = {{\bf{h}}_i}{\bf{h}}_i^{\rm{H}}$, and ${\mathop{\rm rank}\nolimits} \left( {{{\bf{U}}_{\rm{c}}}} \right) = {\mathop{\rm rank}\nolimits} \left( {{{\bf{U}}_m}} \right) = 1$. 
For non-convex constraints (\ref{P1_supp}d), (\ref{P1_supp}e), and (\ref{P1_supp}f), it holds that
\begin{equation}\label{P1_b1}
	\begin{array}{l} \displaystyle
		{\mathop{\rm Tr}\nolimits} \left( {{{\bf{H}}_m}{{\bf{U}}_{\rm{c}}}} \right) - \left( {{e^{\sum\nolimits_{m = 1}^M {{r_{{\rm{c}},m}}} }} - 1} \right)\sum\nolimits_{j = 1}^M {{\mathop{\rm Tr}\nolimits} \left( {{{\bf{H}}_m}{{\bf{U}}_j}} \right)} \\ \displaystyle \qquad\qquad \qquad\qquad\qquad \qquad
		\ge \left( {{e^{\sum\nolimits_{m = 1}^M {{r_{{\rm{c}},m}}} }} - 1} \right){\sigma ^2},
	\end{array}
\end{equation}
\begin{equation}\label{P1_c1}
	{\mathop{\rm Tr}\nolimits} \left( {{{\bf{H}}_m}{{\bf{U}}_{\rm{c}}}} \right) - \left( {{2^{{I_{\rm{c}}}}} - 1} \right)\sum\nolimits_{j = 1}^M {{\mathop{\rm Tr}\nolimits} \left( {{{\bf{H}}_m}{{\bf{U}}_j}} \right)}  \ge \left( {{2^{{I_{\rm{c}}}}} - 1} \right){\sigma ^2},
\end{equation}
and
\begin{equation}\label{P1_d1}
	{\mathop{\rm Tr}\nolimits} \left( {{{\bf{H}}_m}{{\bf{U}}_m}} \right) - \left( {{2^{{I_{\rm{p}}}}} - 1} \right)\sum\nolimits_{j \ne m} {{\mathop{\rm Tr}\nolimits} \left( {{{\bf{H}}_m}{{\bf{U}}_j}} \right)}  \ge \left( {{2^{{I_{\rm{p}}}}} - 1} \right){\sigma ^2}.
\end{equation}
Let ${\bf{c}} = \left\{ {{c_{1,m}},{c_{2,m}},{c_{3,m}},{c_{4,m}},{c_{5,m}},{c_{6,m}}} \right\}$ be a set of auxiliary variables. Then, (\ref{P1_b1}), (\ref{P1_c1}), and (\ref{P1_d1}) can be rewritten as 
\begin{subequations} \label{P1_b2}
	\begin{align}
		&{\mathop{\rm Tr}\nolimits} \left( {{{\bf{H}}_m}{{\bf{U}}_{\rm c}}} \right) \ge {c_{1,m}}{c_{2,m}}\\
		&{c_{1,m}} \ge {{2^{\sum\nolimits_{j = 1}^M {{r_{{\rm{c}},j}}} }} - 1},\\ 
		&{c_{2,m}} \ge \sum\nolimits_{j= 1}^M {{\mathop{\rm Tr}\nolimits} \left( {{{\bf{H}}_m}{{\bf{U}}_j}} \right)}  + {\sigma ^2},
	\end{align}	
\end{subequations}
\begin{subequations} \label{P1_c2}
	\begin{align}
		&{\mathop{\rm Tr}\nolimits} \left( {{{\bf{H}}_m}{{\bf{U}}_{\rm c}}} \right) \ge {c_{3,m}}{c_{4,m}}\\
		&{c_{3,m}} \ge {2^{{I_{\rm{c}}}}} - 1,\\ 
		&{c_{4,m}} \ge \sum\nolimits_{j= 1}^M {{\mathop{\rm Tr}\nolimits} \left( {{{\bf{H}}_m}{{\bf{U}}_j}} \right)}  + {\sigma ^2},
	\end{align}	
\end{subequations}
\begin{subequations} \label{P1_d2}
	\begin{align}
		&{\mathop{\rm Tr}\nolimits} \left( {{{\bf{H}}_m}{{\bf{U}}_m}} \right) \ge {c_{5,m}}{c_{6,m}}\\
		&{c_{5,m}} \ge {2^{{I_{\rm{p}}}}} - 1,\\ 
		&{c_{6,m}} \ge \sum\nolimits_{j \ne m} {{\mathop{\rm Tr}\nolimits} \left( {{{\bf{H}}_m}{{\bf{U}}_j}} \right)}  + {\sigma ^2}.
	\end{align}	
\end{subequations}
Afterwards, by applying the phase rotation at (\ref{P1_b2}a), (\ref{P1_c2}a), and (\ref{P1_d2}a), we obtain
\begin{equation}\label{P1_b3}
	\sqrt {{\rm{Tr}}\left[ {{\rm{Ro}}\left( {{{\bf{H}}_m}{{\bf{U}}_{\rm{c}}}} \right)} \right]}  \ge \sqrt {{c_{1,m}}{c_{2,m}}} ,
\end{equation}
\begin{equation}\label{P1_c3}
	\sqrt {{\rm{Tr}}\left[ {{\rm{Ro}}\left( {{{\bf{H}}_m}{{\bf{U}}_{\rm{c}}}} \right)} \right]}  \ge \sqrt {{c_{3,m}}{c_{4,m}}} ,
\end{equation}
\begin{equation}\label{P1_d3}
	\sqrt {{\mathop{\rm Tr}\nolimits} \left[ {{\rm{Ro}}\left( {{{\bf{H}}_m}{{\bf{U}}_m}} \right)} \right]}  \ge \sqrt {{c_{5,m}}{c_{6,m}}} .
\end{equation}
With the first-order Taylor series expansion employed for the right parts of (\ref{P1_b3}), (\ref{P1_c3}), and (\ref{P1_d3}),  the convex forms of the original non-convex constraints (\ref{P1_supp}d), (\ref{P1_supp}e), and (\ref{P1_supp}f) are presented as
\begin{equation}\label{P1_b4}
	\begin{array}{*{20}{l}} \displaystyle
		{\sqrt {{\rm{Tr}}\left[ {{\rm{Ro}}\left( {{{\bf{H}}_m}{{\bf{U}}_{\rm{c}}}} \right)} \right]}  \ge \sqrt {{c_{1,m,0}}{c_{2,m,0}}} }\\ \displaystyle
		{ + \frac{{\left( {{c_{1,m}} - {c_{1,m,0}}} \right)\sqrt {{c_{1,m,0}}c_{2,m,0}^{ - 1}} }}{2}}\\ \displaystyle
		{ + \frac{{\left( {{c_{2,m}} - {c_{2,m,0}}} \right)\sqrt {c_{1,m,0}^{ - 1}{c_{2,m,0}}} }}{2} +o\left( n \right),}
	\end{array}
\end{equation}
\begin{equation}\label{P1_c4} 
	\begin{array}{*{20}{l}} \displaystyle
		{\sqrt {{\rm{Tr}}\left[ {{\rm{Ro}}\left( {{{\bf{H}}_m}{{\bf{U}}_{\rm{c}}}} \right)} \right]}  \ge \sqrt {{c_{3,m,0}}{c_{4,m,0}}} }\\ \displaystyle
		{ + \frac{{\left( {{c_{3,m}} - {c_{3,m,0}}} \right)\sqrt {{c_{3,m,0}}c_{4,m,0}^{ - 1}} }}{2}}\\ \displaystyle
		{ + \frac{{\left( {{c_{4,m}} - {c_{4,m,0}}} \right)\sqrt {c_{3,m,0}^{ - 1}{c_{4,m,0}}} }}{2} + o\left( n \right),}
	\end{array}
\end{equation}
and
\begin{equation}\label{P1_d4}
	\begin{array}{*{20}{l}}  \displaystyle
		{\sqrt {{\rm{Tr}}\left[ {{\rm{Ro}}\left( {{{\bf{H}}_m}{{\bf{U}}_m}} \right)} \right]}  \ge \sqrt {{c_{5,m,0}}{c_{6,m,0}}} }\\  \displaystyle
		{ + \frac{{\left( {{c_{5,m}} - {c_{5,m,0}}} \right)\sqrt {{c_{5,m,0}}c_{6,m,0}^{ - 1}} }}{2}{\rm{ }}}\\  \displaystyle
		{ + \frac{{\left( {{c_{6,m}} - {c_{6,m,0}}} \right)\sqrt {c_{5,m,0}^{ - 1}{c_{6,m,0}}} }}{2} + o\left( n \right),}
	\end{array}
\end{equation}
where $c_{i,m,0}$ is the specific value of the corresponding auxiliary variable $c_{i,m}$, $i \in \left\{ {1,2,3,4,5,6} \right\}$.
Let ${o}\left( n \right)$ denote the higher-order terms in the Taylor series expansion. For analytical tractability, these terms are omitted in subsequent derivations. Formally, the approximation retains only the first-order components.
With the users' CSI uncertainty provided in (\ref{ICSI_1}) - (\ref{ICSI_4}) taken into account, (\ref{P1_b4}), (\ref{P1_c4}), and (\ref{P1_d4}) are further given by
\begin{equation}\label{P1_b5_1}
	\begin{array}{*{20}{l}} \displaystyle
		{\sqrt {{\rm{minTr}}\left[ {{\rm{Ro}}\left( {{{\bf{H}}_m}{{\bf{U}}_{\rm{c}}}} \right)} \right]}  \ge \sqrt {{c_{1,m,0}}{c_{2,m,0}}} }\\ \displaystyle
		{ + \frac{{\left( {{c_{1,m}} - {c_{1,m,0}}} \right)\sqrt {{c_{1,m,0}}c_{2,m,0}^{ - 1}} }}{2}}\\ \displaystyle
		{ + \frac{{\left( {{c_{2,m}} - {c_{2,m,0}}} \right)\sqrt {c_{1,m,0}^{ - 1}{c_{2,m,0}}} }}{2} +o\left( n \right),}
	\end{array}
\end{equation}
\begin{equation}\label{P1_c5_1}
	\begin{array}{*{20}{l}} \displaystyle
		{\sqrt {{\rm{minTr}}\left[ {{\rm{Ro}}\left( {{{\bf{H}}_m}{{\bf{U}}_{\rm{c}}}} \right)} \right]}  \ge \sqrt {{c_{3,m,0}}{c_{4,m,0}}} }\\ \displaystyle
		{ + \frac{{\left( {{c_{3,m}} - {c_{3,m,0}}} \right)\sqrt {{c_{3,m,0}}c_{4,m,0}^{ - 1}} }}{2}}\\ \displaystyle
		{ + \frac{{\left( {{c_{4,m}} - {c_{4,m,0}}} \right)\sqrt {c_{3,m,0}^{ - 1}{c_{4,m,0}}} }}{2} + o\left( n \right),}
	\end{array}
\end{equation}
and
\begin{equation}\label{P1_d5_1}
	\begin{array}{*{20}{l}}  \displaystyle
		{\sqrt {{\rm{minTr}}\left[ {{\rm{Ro}}\left( {{{\bf{H}}_m}{{\bf{U}}_m}} \right)} \right]}  \ge \sqrt {{c_{5,m,0}}{c_{6,m,0}}} }\\  \displaystyle
		{ + \frac{{\left( {{c_{5,m}} - {c_{5,m,0}}} \right)\sqrt {{c_{5,m,0}}c_{6,m,0}^{ - 1}} }}{2}{\rm{ }}}\\  \displaystyle
		{ + \frac{{\left( {{c_{6,m}} - {c_{6,m,0}}} \right)\sqrt {c_{5,m,0}^{ - 1}{c_{6,m,0}}} }}{2} + o\left( n \right),}
	\end{array}
\end{equation}
For the constraint (\ref{P1}f), let ${\gamma _{\min }} = \min \left\{ {{\gamma _i},{\gamma _2}, \cdots ,{\gamma _K}} \right\}$ and ${\gamma _{\max }} = \max \left\{ {{\gamma _i},{\gamma _2}, \cdots ,{\gamma _K}} \right\}$, which leads to ${\gamma _{\max }} - {\gamma _{\min }} \le \rho_0 $ which implies that ${\gamma _{\max }} \le \rho_0  + {\gamma _{\min }}$. Hence, we get
\begin{equation}\label{gamma_3}
	{\gamma _{\min }} \le {\gamma _i} \le {\gamma _{\min }} + \rho_0 ,i \in \left\{ {1,2, \cdots ,K} \right\}.
\end{equation}
By applying (\ref{gamma_echo}) in (\ref{gamma_3}), we obtain 
\begin{equation}\label{gamma_4}
	\displaystyle \frac{{{\gamma _{\min }}}}{{{\gamma _{o,k}}}} \le {o_k} \le \frac{{\left( {{\gamma _{\min }} + {\rho _0}} \right)}}{{{\gamma _{o,k}}}}.
\end{equation}
In (\ref{gamma_4}), the parameter $\rho_0$ determines the upper bound of the achievable sensing performance improvement margin for the $k$-th target beyond its basic sensing requirements, while serving as a critical tradeoff parameter between sensing and communication. This parameter can be dynamically adjusted in real time to adopt to network conditions. For instance, during network congestion, $\rho_0$ is reduced to relax sensing demands, thereby prioritizing communication reliability.
Therefore, we obtain the rewritten optimization problem P3 shown in (\ref{P2}), implying that the Theorem \ref{theorem_1} gets proved.
 \end{appendices}

\bibliographystyle{IEEEtran}
\bibliography{ref_MTMU_RSMA_ISAC}

\end{document}